\documentclass{article}
\usepackage{arxiv}
\usepackage{amsmath}
\usepackage{mathtools}
\usepackage{amsfonts}
\usepackage{amsthm}
\usepackage{booktabs}

\usepackage{makecell}
\newcommand{\W}{X^*}
\newcommand{\w}{x^*}
\newcommand{\vw}{v}
\DeclareMathOperator*{\argmax}{arg\,max}
\DeclareMathOperator{\cov}{cov}
\DeclareMathOperator{\Tr}{Tr}
\newtheorem{thm}{Result}

\title{Measurement error and precision medicine: error-prone tailoring covariates in dynamic treatment regimes.}

\date{}

\author{
	Dylan Spicker\\
	Statistics and Actuarial Science\\ 
	University of Waterloo\\ 
	Waterloo, Ontario, N2L 3G1\\
	\texttt{dylan.spicker@uwaterloo.ca} \\
	\And
	Michael Wallace \\
	Statistics and Actuarial Science\\ 
	University of Waterloo\\ 
	Waterloo, Ontario, N2L 3G1
}

\begin{document}

\maketitle

\begin{abstract}
Precision medicine incorporates patient-level covariates to tailor treatment decisions, seeking to improve outcomes. In longitudinal studies with time-varying covariates and sequential treatment decisions, precision medicine can be formalized with dynamic treatment regimes (DTRs): sequences of covariate-dependent treatment rules. To date, the precision medicine literature has not addressed a ubiquitous concern in health research - measurement error - where observed data deviate from the truth. We discuss the consequences of ignoring measurement error in the context of DTRs, focusing on challenges unique to precision medicine. We show - through simulation and theoretical results - that relatively simple measurement error correction techniques can lead to substantial improvements over uncorrected analyses, and apply these findings to the Sequenced Treatment Alternatives to Relieve Depression (STAR*D) study.
\end{abstract}

\section{Introduction}
Precision medicine is a framework in which medical treatment decisions are based on patient-level data. At its core, precision medicine aims to `treat patients, not diseases', reflecting the principle that the best treatment decision is informed by all relevant, available data on the patient, not solely their diagnosis. This can manifest in the simple case of a single treatment decision (the \textit{one-stage} setting), but can be readily generalized to longitudinal treatment regimes where all available patient-level data (both present and past) can inform treatment decisions. One way of codifying such a process in the precision medicine framework is through the use of Dynamic Treatment Regimes (DTRs): sequences of decision rules tailored to patient-level covariates. Precision medicine in general, and DTRs in particular, have received a great deal of research attention in recent years.\cite{Med_Hamburg,Med_Katsios,Med_Smalley,DTR_text}

A key focus of the DTR framework is estimating the \textit{optimal} sequence of treatment decisions that maximize an expected outcome, conditional on available patient-level data at each decision point. This may be a simple rule based on a single covariate (such as ``prescribe treatment if the patient consumes fewer than $1300$ calories a day''), or may be a highly complex set of treatment decisions which depend on many factors. Finding optimal treatment rules can be especially challenging in the observational data setting, wherein observed treatments may themselves be informed by patient-level information. Estimation of such rules has received considerable attention in the biostatistical literature, with the development of numerous estimation procedures.\cite{DTR_text,Wallace_Overview}

Measurement error refers to any process through which observable data do not equal the true underlying values of interest.\cite{YiBook} Common examples include blood pressure (typically elevated in clinical settings) \cite{BP_ME} or self-reported caloric intake.\cite{Calories_ME} While measurement error may arise through a variety of mathematical mechanisms, the underlying concern is that analyses which do not account for such error may produce unpredictable and unreliable conclusions. These so-called \textit{naive analyses}, along with many relevant correction methods, have been widely studied in both linear \cite{FullerBook} and non-linear \cite{CarrollBook} models.

Despite the abundance of literature surrounding both DTRs and measurement error, there has yet to be a substantial attempt to assess the estimation of the former in the presence of the latter. Precision medicine encompasses both estimation and inference (in studying treatment effects) but also prediction (in applying decision rules to future patients). While estimation and inference have received considerable attention in the measurement error literature (outside of the context of precision medicine), prediction has not. Some argue for the use of error-prone variables to predict an outcome of interest directly. This argument advocates exploiting the dependence between our error-prone variable and the outcome, ignoring the underlying causal relationship between the true variable and the outcome. Then, if there is a sufficiently strong relationship between the outcome and observed measure, prediction may remain valid. This argument is correct for some models,\cite{YiBook} however, it is not always clear when it is valid. It has been shown that even in a standard regression setting, it may be necessary to correct for error to ensure valid prediction.\cite{Pred_Butts}

Whether the goal is to assess the efficacy of a treatment regime, or to aid in future treatment decisions, the literature in standard modeling contexts gives us reason to believe that many of the same concerns will exist in precision medicine, and DTRs specifically. We will investigate these issues, establish when and how precision medicine analyses may be affected, and investigate - through simulation and theoretical results - straightforward measurement error correction techniques within the DTR framework. We further apply our methods to analyze the Sequenced Treatment Alternatives to Relieve Depression (STAR*D) study.\cite{Fava2003,Rush2004}

\section{Methodology}
A wide variety of methods are available both for DTR estimation and measurement error correction. Common techniques for the former include Q-learning \cite{Q_Sutton,Q_Watkins} and inverse probability weighting,\cite{horvitz1952} along with the more complex and robust G-estimation \cite{G_Joffe,G_Robbins} and augmented inverse probability weighting techniques.\cite{AIPW} Dynamic weighted ordinary least squares (dWOLS) offers a compromise between these broad classes of approach, offering robustness to model misspecification while maintaining straightforward implementation.\cite{dwols} Classification-type approaches, such as outcome weighted learning, have also recently grown in popularity.\cite{Classif_Zhang,Classif_Zhao_2012}

The measurement error literature boasts a similar range of options with a familiar trade-off between simplicity of theory and theoretical properties. For example, regression calibration \cite{Carroll1990,Gleser1990} and simulation extrapolation (SIMEX) \cite{Carroll_SIMEX,Stefanski_SIMEX} are both general methods, which make few assumptions regarding the underlying data. This results in comparatively simple estimators which are consistent in linear models, but which can make no general consistency guarantees in non-linear models.\cite{CarrollBook} In contrast, methods such as the conditional score\cite{Cond_Score_1,Cond_Score_2}, under correct distributional assumptions, and the corrected score\cite{Corrected_Score_1,Corrected_Score_2} offer consistent estimators for a larger class of models, at the cost of more complex implementation.

As the first substantive work on the interface between DTRs and measurement error, we will limit our focus to methodologies that afford straightforward implementation. DTR estimation will be carried out via dWOLS,\cite{dwols} whose regression-based implementation is complemented by the measurement error correction method of regression calibration.\cite{Carroll1990,Gleser1990}

To establish notation, and the modeling framework upon which our methodology relies, in this section we will introduce the specifics of a one-stage, error-free DTR, and discuss the estimation procedure using dWOLS. We will then extend to the multistage case. Finally, we will discuss regression calibration generally, clarifying the specific corrections that we will use.

\subsection{DTRs and dWOLS}\label{DTR_subsec}
In a one-stage DTR we make a single binary treatment decision ($A$) per patient. We take $A = a \in \{0,1\}$ to denote some binary treatment option, such as standard treatment ($a=0$) compared to an intervention ($a=1$). We are concerned with an outcome variable $Y$, chosen such that larger values are preferred. Patient information available immediately prior to the decision being made is denoted $\mathbf{X}$. The optimal DTR will then take $\mathbf{X}=\mathbf{x}$ as input for a patient and return $A=a^{\text{opt}}$, such that $Y$ is maximized in expectation. Formally, $a^\text{opt} = \argmax_{a} E[Y|A=a, \mathbf{X}=\mathbf{x}; \beta, \psi]$, where the mean is modeled as
\begin{equation}
	E[Y|A=a, \mathbf{X}=\mathbf{x}; \beta, \psi]= f(\mathbf{x}_{\beta}; \beta) + \gamma(\mathbf{x}_{\psi}, a; \psi) \label{OneStageDTR}
\end{equation}
\noindent with $\mathbf{x}_{\beta}$ and $\mathbf{x}_{\psi}$ representing two (possibly identical) subsets of the covariates $\mathbf{x}$, and $\beta$ and $\psi$ are model parameters. We often take $\mathbf{x} = (1, \mathbf{x})$ to allow for a baseline effect (or intercept) to be captured in $f$. 

This mean decomposition includes a component $f(\mathbf{x}_{\beta};\beta)$ which does not depend on the treatment, and a component $\gamma(\mathbf{x}_{\psi}, a; \psi)$ which captures the effect of treatment. These are the so-called \textit{treatment-free} and \textit{blip} components, respectively. The combination of the treatment-free and blip models (as outlined in Equation (\ref{OneStageDTR})) is referred to as the \textit{outcome model}. The treatment decision impacts the outcome only through the blip. As such, estimation of the optimal DTR is equivalent to finding the decision rule which maximizes $\gamma(\mathbf{x}_{\psi}, a; \psi)$. We therefore only need to estimate the $\psi$ terms correctly to determine the optimal DTR.

Often, we take $\gamma(\mathbf{x}_\psi, a; \psi) = a\psi'\mathbf{x}_\psi$ to be a linear function of the covariates multiplied by the treatment indicator, meaning $a^\text{opt} = 1$ if $\psi'\mathbf{x}_\psi > 0$ and $a^\text{opt} = 0$ otherwise. If we correctly specify the full outcome model then standard regression procedures may be applied. However, as our treatment decision does not depend on the treatment-free component directly, we may wish to seek methodology that does not depend on its correct specification in full. For example, even with the treatment-free model misspecified, we could nevertheless proceed with correct specification of the blip if $a$ and $\mathbf{x}$ were independent, but this is seldom a reasonable assumption in our observational setting (where treatment decisions may be made based on patient-level data).

dWOLS, along with some other methods such as the aforementioned G-estimation \cite{G_Robbins} and augmented inverse probability of treatment weighting,\cite{AIPW} account for this by requiring the specification of a \textit{treatment model}, modeling the probability of receiving the intervention given the individual's covariates. In dWOLS, this allows the calculation of patient-level weights, which we denote $\vw(a, \mathbf{x})$ for a patient with covariates $\mathbf{x}$ receiving treatment $a$. These weights are designed to `balance' the covariates. Any weights which satisfy
\begin{equation}
	\pi(\mathbf{x})v(1, \mathbf{x}) = (1 - \pi(\mathbf{x}))v(0, \mathbf{x}), \label{dwols_weights}
\end{equation}
where $\pi(\mathbf{x}) = P(A=1|\mathbf{X}=\mathbf{x})$, will suffice for balance. The use of $\vw(a, \mathbf{x}) = |a - \pi(\mathbf{x})|$ is recommended.\cite{dwols}

In the one-stage setting, dWOLS is simply a weighted ordinary least squares regression with weights satisfying Equation (\ref{dwols_weights}) and an outcome model structure as indicated by Equation (\ref{OneStageDTR}). The resulting estimators for $\psi$ are then \textit{doubly robust}: as long as the blip model is correctly specified, and at least one of the treatment or treatment-free models is correctly specified, the estimators for $\psi$ will be consistent.


These methods easily extend to multistage processes. A $K$-stage DTR will have $K$ total treatment decisions made, which we index by $j$. We wish to estimate the optimal decision for stage $j$, given all information available immediately prior to the decision. We now subscript the covariate vector and the treatment decision, $\mathbf{x}_j$ and $a_j$, to denote the measurements taken and the observed decision at stage $j$, respectively. We use over- and under-line notation to refer to the past and future respectively, so that (for example) $\overline{\mathbf{x}}_j = (\mathbf{x}_1, \hdots, \mathbf{x}_j)$ and $\underline{\mathbf{a}}_{j+1} = (a_{j + 1}, \hdots, a_K)$. Finally, for notational convenience, we define a variable to represent the patient's history prior to the stage $j$ treatment decision: $\mathbf{h}_j = (\overline{\mathbf{x}}_j, \overline{\mathbf{a}}_{j-1})$. We now expand Equation (\ref{OneStageDTR}), given the above notation, to
\begin{equation}
	E[Y|\mathbf{H}=\mathbf{h}; \beta, \psi] = \sum_{j=1}^K \left\{f_j(\mathbf{h}_j^\beta; \beta_j) + \gamma_j(\mathbf{h}_j^\psi, a_j; \psi_j)\right\}
\end{equation}
\noindent where $\mathbf{h}_j^\beta$ and $\mathbf{h}_j^\psi$, are (possibly identical) subsets of the history vector $\mathbf{h}_j$. We take $f_j$ to be the treatment-free model for stage $j$, specifying the impact of $\mathbf{h}_j^\beta$ on the outcome. This impact is not mediated by treatment. Conversely, $\gamma_j$ is the blip model for stage $j$, which indicates the impact of $\mathbf{h}_j^\psi$ on the outcome. This effect is mediated by treatment ($a_j$).

The stage $j$ blip function in this multistage setting is defined as the marginal impact of a patient receiving treatment $a_j$, compared to a patient who received standard treatment at stage $j$ ($a_j = 0)$, with an identical history and who goes on to receive optimal, though not necessarily identical, treatment in the future. That is, $\gamma_j(\mathbf{h}_j, a_j; \psi_j) = E[Y^{\overline{\mathbf{a}}_{j}, \underline{\mathbf{a}}_{j+1}^{\text{opt}}} - Y^{\overline{\mathbf{a}}_{j-1}, 0,\underline{\mathbf{a}}_{j+1}^{\text{opt}}}| \mathbf{H}_j = \mathbf{h}_j]$, where $Y^{a^\dagger}$ refers to the \textit{counterfactual} outcome $Y$, which is potentially unobserved, under a treatment regime specified by $a^\dagger$. 

There is an alternative formulation for the outcome model, which provides a more intuitive characterization for interpreting dWOLS. The mean outcome can be defined such that \begin{equation} E[Y|\mathbf{H}=\mathbf{h}] = E[Y^\text{opt}|\mathbf{H}=\mathbf{h}] - \sum_{j=1}^K \mu_j(\mathbf{h}_j, a_j; \psi_j) \label{eq:outcome_regret}, \end{equation} where $Y^\text{opt}$ gives the theoretically optimal outcome (under the optimal DTR) and the $\mu_j$ constitute penalty terms for non-optimal treatment. Here $\mu_j$ are referred to as \textit{regrets}, and are the contrast in outcomes between a patient who receives optimal treatment at stage $j$, and the same patient receiving treatment $a_j$ at stage $j$, assuming the patient is treated optimally thereafter. Formed this way, the observed outcome is equal to the optimal outcome less all negative effects deriving from suboptimal treatment. The regrets may be expressed in terms of the blips, namely $\mu_j(\mathbf{h}_j, a_j) = \gamma_j(\mathbf{h}_j, a_j^\text{opt}) - \gamma_j(\mathbf{h}_j, a_j)$. 

There is a recursive nature to the multistage DTR analysis as the treatment decision at stage $j$ impacts all future decisions. In dWOLS we begin by analyzing the \textit{final} stage of treatment, then work backwards, at each stage generating a \textit{pseudo-outcome} which removes the effects of future treatment from the outcome. Letting $\tilde{y}_K = Y$, we define the $j$-th pseudo-outcome as $\tilde{y}_j = \tilde{y}_{j+1} + \left(\gamma_j(\mathbf{h}_j, a_j^\text{opt}) - \gamma_j(\mathbf{h}_j, a_j)\right) = \tilde{y}_{j+1} + \mu_j(\mathbf{h}_j)$. That is, we `add back' the stage $j$ regret, effectively removing it from the outcome. This allows the pseudo-outcome at stage $j$ to be interpreted as the outcome for a patient who receives their particular regime up to stage $j$, and then is optimally treated afterwards. We could continue to use the blip formulation of the outcome model. In this case we once again take $\tilde{y}_K' = Y$, and then define $\tilde{y}_j' = \tilde{y}_{j+1}' - \gamma_j(\mathbf{h}_j, a_j)$. In practice, the regret setup is more commonly implemented, and we will continue to use it (unless otherwise stated).

Estimation using dWOLS in the multistage setting then follows a three step procedure. First, define weights for each stage $\vw_j$, which satisfy Equation (\ref{dwols_weights}), using $\overline{\mathbf{x}}_j$ as the covariate. Second, starting at stage $K$, and working iteratively backwards, solve the weighted regression of $\tilde{y}_j$ on the patient history $\mathbf{h}_j$. Third, define $\tilde{y}_{j-1}$, and repeat. If the blip and at least one of the treatment or treatment-free models are correctly specified at each stage $j$, this process will lead to consistent estimators for all $\psi_j$.

In order to estimate a DTR we further require two assumptions on our data, whether they are randomized or observational. First, we make the stable unit treatment value assumption (SUTVA).\cite{SUTVA} SUTVA requires that a patient's outcome is not influenced by another patient's treatment assignment. This is typically reasonable, though may be violated, for instance, when the intervention is a vaccine and the effects of herd immunity influence all observed outcomes. Second, we make the no unmeasured confounders, or sequential ignorability, assumption.\cite{NUC} No unmeasured confounders requires that all common causes of treatment (at each stage $j$) and future potential covariates or outcomes must be measured in the history. That is, conditional on the available history, treatment must be independent of future potential covariates and outcome. While this assumption will typically hold in randomized studies, it is untestable in the observational framework, and so must be carefully validated based on the applicable subject matter. We will make these assumptions for the remainder of our discussion.

\subsection{Measurement Error and Regression Calibration}
In order to correct for measurement error we need to make assumptions regarding the structure of the error. If we take $X$ to be the true covariate and $\W$ to be an error-prone observation of $X$, then we assume some form for $\W = g(X,U)$. Here $U$ is the random noise which induces the error. While any specific application may suggest a particular form for the error model ($g$), two commonly used models are the \textit{classical additive}, and the \textit{multiplicative} error models. In the classical additive model we take $g(X,U) = X + U$, and assume that $E[U] = 0$, and that $U$ has constant covariance, not depending on $X$, given by $\cov(U) = \Sigma_U$. Moreover, we assume that $U$ is independent of (or sometimes uncorrelated with) both $X$, and any other covariates we observe (without error) $Z$. In the multiplicative model we take $g(X,U) = XU$, and assume that $E[U] = 1$, and again, $\cov(U) = \Sigma_U$. We make the same independence assumptions. Note, that while the notation above tends to imply that $X$ (and as such $U$) is a scalar, the same models can be extended to vector-valued covariates; we make no distinction, and have selected notation for simplicity of exposition. Both the additive and multiplicative models provide \textit{unbiased} measurements of $X$, in the sense that $E[\W|X,Z] = X$. When we have an outcome of interest, $Y$, we also tend to classify error as either \textit{differential} or \textit{non-differential}. Non-differential error refers to the case where, given $\{X,Z\}$, our outcome $Y$ is conditionally independent of $\W$. Errors may be differential if, for instance, measurements are taken subsequent to the outcome being observed (such as a cancer diagnosis affecting how a patient responds to questions about their historical smoking habits). This is seldom the case in our framework.

Error correction techniques (such as our choice of regression calibration) typically require additional data beyond what is used in standard inferential procedures to learn about the size (and structure) of the error. These may come in the form of validation, replicate, or instrumental data. Validation data consist of a subsample of the observations where both the true and error-prone observations are available. Replicate data consist of repeated measurements of the error-prone covariates for some subset of the individuals. Instrumental data, also known as instrumental variables (IV), refer to additional covariates that are related to the true values, but which are (typically) uncorrelated with both the error observed in the covariate, and the variability in the model after accounting for the true covariates.\cite{CarrollBook} An IV, $T$ is called unbiased for $X$, if $E[T|X,Z]=X$. Replicate measurements can be viewed as a specific type of unbiased instrumental variables. While validation data are typically considered ideal, they are often unavailable in practice. Instead, we focus on the use of unbiased IVs, referred to as error-prone proxies, including replicate measurements.

The premise of regression calibration is to replace the unobserved $X$ in our models with an estimated $\widehat{X} = E[X|Z,\W]$. We then proceed with standard analysis on the predicted values, adjusting the standard errors as needed. Consider a single patient, with $k$ unbiased proxies of the true covariate, denoted $\W_1, \hdots, \W_k$. A common procedure for determining $\widehat{X} = E[X|Z,\W]$ is taking the so-called \textit{best linear unbiased prediction} (BLUP) approach, which involves approximating the conditional mean as a linear equation.\cite{Carroll1990,Gleser1990} We use a plug-in estimator for the theoretical BLUP, given by \begin{equation}
\widehat{X} = \mu_X + \begin{bmatrix}\Sigma_{XX^*} & \Sigma_{XZ}\end{bmatrix}\begin{bmatrix}
        \Sigma_{X^*X^*} & \Sigma_{X^*Z} \\
        \Sigma_{ZX^*} & \Sigma_{ZZ}
    \end{bmatrix}^{-1}\begin{bmatrix}
        X^* - \mu_X \\
        Z - \mu_Z
    \end{bmatrix}.\label{eq:simplified_form}
\end{equation} Regression calibration was originally proposed for a fairly general class of additive error models, with independent and identically distributed replicate measurements. \cite{Carroll1990} When we have identically distributed replicate measures, it is sensible to take the mean of $\W_1,\hdots,\W_k$ for use as $\W$ in this equation. When the error models differ between proxies, it seems unlikely that the most effective way to combine the proxies is a simple mean. Intuitively, the observations which are less disturbed by error ought to contribute more to $\W$. Our estimators are a generalization of those typically used when replicates are available, \cite{CarrollBook} where we take $\W = \sum_{j=1}^k \delta_j\W_j$, for a set of weights $\delta_j$, choosen either for the interpretation of the estimator or to reduce variability. The specific details of the implementation of our estimators are contained in Appendix \ref{apdx:RC_details}. 

Applying this correction, when errors are additive, results in consistent estimators in linear models. In non-linear models, some authors have described the estimators as ``nearly consistent''.\cite{MEMInteractions} By this they mean that, while we cannot guarantee consistency of all parameters in the models, we can often make claims about reduced bias or consistency of some parameters. For instance, in log-linear models, regression calibration estimators will consistently estimate the slope parameters while inconsistently estimating the intercept.\cite{CarrollBook} Of greatest concern for our work, outside of linear models, is the utility of regression calibration for logistic regression models. Here, the phrasing ``nearly consistent'' is taken to mean that the correction, in many situations, provides a great reduction in the bias of the parameter estimates. Further, when the main concern with the fitted logistic regression is the probability estimates, regression calibration provides a reasonable approximation to the truth, when neither the slope parameters, nor the conditional variance of $X$ given $\{\W,Z\}$ are too large. Specifically, \[P(A=1| Z, \W) \approx H\left[\frac{\alpha_0 + \alpha_X'\widehat{X} + \alpha_Z'Z}{(1 + \alpha_X'\Sigma_{X|Z,\W}\alpha_X/1.7^2)^{1/2}}\right],\] where $H(x) = (1 + \exp(-x))^{-1}$ is the inverse-logit (expit) function, $\alpha$ are the parameters of the logistic regression, and $\Sigma_{X|Z,\W}$ gives the conditional covariance of $X$ given $\{Z,\W\}$.\cite{CarrollBook} In general, the denominator will attenuate the estimates for the model parameters. When the denominator of this approximation is close to $1$, however, this attenuation will be small, and the estimator using regression calibration takes an approximately correct form.

While the corrections we consider are motivated by non-differential classical additive error models, a wider class of error models can be accommodated by the methods we introduce. Any non-differential error model, with (1) more than one unbiased proxy available for $X$, (2) uncorrelated errors between proxies (that is, the covariance between any two unbiased measurements equals the covariance of $X$, $\Sigma_{X}$), and (3) the covariance of all of the unbiased proxies given as $\Sigma_{X} + M$ for suitable constant $M$, can be interchanged if necessary. This would include the multiplicative error model introduced above. However, when errors are multiplicative, it is generally advisable that transformations are applied to the measured values, to use a scale on which the errors are additive. \cite{Multiplicative_Transformations} While the existing asymptotic theory, and most implementations of regression calibration correction, rely on the additive structure, the modified estimators we present are computable under this wider class of models. In our simulation studies (Section \ref{sec:simulations}) we demonstrate that the corrections perform adequately under slight deviations from additivity. Still, if there is good reason to believe that all proxies are subject to non-additive error, which is transformable to an additive model, we would advise that analysts make these transformations, as asymptotic justifications for the presented methods on those models are lacking.

\section{Measurement Error in Dynamic Treatment Regimes}\label{ME_in_DTRs}


With little work to date concerning measurement error in DTRs and precision medicine, we shall limit attention to the case of errors in patient-level covariates, assuming treatments and outcomes are measured correctly. In this section we will illustrate that, in addition to concerns arising from measurement error that are common to traditional modeling settings, there are considerations which are unique to the structure of DTRs. We discuss these considerations alongside various theoretical observations. 

First, we motivate the construction of valid estimators in the one-stage setting, providing a theoretical guarantee of sample covariate balance in the presence of measurement error. Next, we illustrate how these estimators can be extended to the multistage setting, paying specific attention to the estimation of pseudo-outcomes. We then discuss how confidence intervals may be constructed using a modified bootstrap procedure. Finally, we focus on the process of determining optimal treatments in the future, reframing the estimation as a prediction problem, and discuss the merits of error correction in this context.


\subsection{Blip Parameter Estimation}\label{p_estimation}
To estimate a DTR with dWOLS, we must specify the outcome and treatment models. We assume that a biological process (or similar) relates the true covariate values $X$ to the outcome $Y$, whereas treatment decisions can only be made based on the error-prone observed values $\W$. This structure is shown graphically in Figure \ref{DAG_DTR_Me}. The treatment decision may be based on a single observed proxy, for instance $\W_1$, or on some combination of these proxies, for instance $\W = \frac{1}{k}\sum_{i=1}^k \W_i$.

\begin{figure}
 \centering 
  \includegraphics{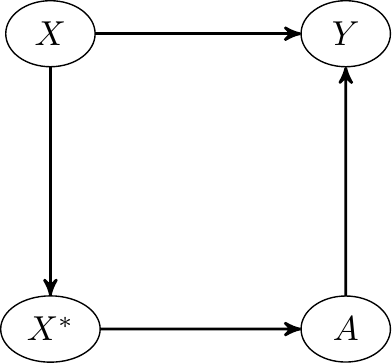}
 \caption{A directed acyclic graph representing the assumed impact of measurement error in the DTR setting. Here $X$ refers to error-free (unobservable) covariates, $X^*$ refers to some combination of the error-prone proxies, $A$ is a binary treatment indicator, and $Y$ is a numeric outcome of interest.}
 \label{DAG_DTR_Me}
\end{figure}

Recall that dWOLS produces doubly robust estimators through the principle of covariate balance. If our treatment-free model is misspecified, then correct specification of the treatment model will induce covariate balance in $X$, (that is $E[X|A=1] = E[X|A=0]$ in the weighted data set), leading to consistent estimation of the blip parameters.\cite{dwols} In the error-prone setting, if we employ regression calibration using $\widehat{X}$ in our outcome models, we wish to induce covariate balance not between $X$ and $A$, but between $\widehat{X}$ and $A$. Following the proof by Wallace and Moodie which justifies the choice of weights,\cite{dwols} we might intuitively speculate that any weights which satisfy $\pi(\widehat{X})\vw(1,\widehat{X}) = (1-\pi(\widehat{X}))\vw(0,\widehat{X})$ will induce covariate balance in $\widehat{X}$. In the error-free setting weights of the form $\vw(a,x) = |a - E[A|X = x]|$ are recommended, suggesting the use of weights of the form $\vw(a, \w) = |a - E[A|\widehat{X} = \w]|$ may prove suitable. We observe the following Result.

\begin{thm}\label{covariate_balance}
	Let $\mathbf{\W} = (\W_1,\hdots,\W_k)$ be observed as a set of unbiased proxies for $X$, and denote the regression calibration estimates of $X$, based on $\mathbf{\W}$, by $\widehat{X}$. Take $P(A=1|L) = \text{H}(\alpha'L)$, where $L$ is any tailoring covariate (vector), which may (but does not necessarily) contain $X$. Using the weights $\vw(A, \widehat{X}) = |A - \hat{P}(A=1|\widehat{X})|$, where $\hat{P}(A=1|\widehat{X})$ is estimated through a simple logistic regression, the weighted sample means, conditional on $A$, will be equivalent. Given the above setup,
    \begin{align*}
    	\frac{\sum_{i=1}^n \vw(1,\w_{ri})A_i\w_{ri}}{\sum_{i=1}^n \vw(1,\w_{ri})A_i} = \frac{\sum_{i=1}^n \vw(0,\w_{ri})(1 - A_i)\w_{ri}}{\sum_{i=1}^n \vw(0,\w_{ri})(1 - A_i)}.
    \end{align*}
\textbf{Proof of Result 1:} See appendix.
\end{thm}

Result \ref{covariate_balance} states that the estimation procedure alone ensures that, within the sample, the weighted means are equivalent between observations with $A=1$ and with $A=0$, regardless of the true underlying treatment model. The idea of using \textit{sample balance} as a small-sample proxy for true balance has been explored in a traditional balancing score setting.\cite{Rosenbaum1983} We will show through simulations that, in many situations, this sample balance suffices to maintain the double robustness of dWOLS. 

When introducing regression calibration we discussed that, in general, the parameters are not consistently estimated in a logistic regression. This discussion is less directly relevant to our present scenario. The reason is that, under our assumed model, the ``true'' covariates used to inform treatment are the observed covariates, which by assumption are $X^*$. While $X^*$ is error-prone with respect to the underlying value of interest, $X$, and as such the outcome model, it is not error-prone with respect to the treatment model, where decisions are informed using $X^*$. As a result, when applying regression calibration to the treatment model, we are using $\widehat{X}$ in place of $X^*$, not in place of $X$. This remains an approximation -- one which is shown in simulations to adequately induce balance -- but not the standard approximation discussed in the literature.

\subsection{Multistage DTRs with Error}\label{multistage_subsec}
Having established an approach for blip parameter estimation in the single-stage problem, we now consider the multistage case, which further requires the estimation of the pseudo-outcomes, $\tilde{y}_i$. If using the regret formulation, then to estimate the pseudo-outcome we must estimate $\hat{a}_j^\text{opt}$, as well as the blip function itself. Consider, for notational ease, a two-stage DTR, which has a linear specification for $\gamma_2(x_2, a_2) = a_2(\psi_{20} + \psi_{21}'x_2)$. Then $a_2^\text{opt} = I(\psi_{20} + \psi_{21}'x_2 > 0)$. In the error-free case, we have $\hat{a}_j^\text{opt} = I(\widehat{\psi}_{20} + \widehat{\psi}_{21}'x_2 > 0)$ rendering $\widehat{\tilde{y}}_1 = Y - \mu_1 - \mu_2 + (\hat{a}_2^\text{opt} - a_2)(\widehat{\psi}_{20} + \widehat{\psi}_{21}'x_2)$. Under the condition that $\widehat{\psi}_{20} = \psi_{20}$ and $\widehat{\psi}_{21} = \psi_{21}$, then $\hat{a}_2^\text{opt} = a_2^\text{opt}$ and $\widehat{\mu}_2 = (\hat{a}_2^\text{opt} - a_2)(\widehat{\psi}_{20} + \widehat{\psi}_{21}'x_2)$, simplifying the estimated pseudo-outcome to $\widehat{\tilde{y}}_1 = Y - \mu_1$, which is the same as the theoretical quantity $\tilde{y}_1$. 

This simplification will not (necessarily) occur in the error-prone case for two reasons. First, even if the $\widehat{\psi}$ are correctly estimated, the use of $\widehat{X}_2$ in place of $X_2$ will result in a residual term between the blip functions. Second, the estimated optimal treatment $\hat{a}_2^\text{opt}$ may differ from the true optimal treatment. Assuming that we have $\widehat{\psi}_{20} = \psi_{20}$ and $\widehat{\psi}_{21} = \psi_{21}$, then we have \begin{equation}\hat{\mu}_2 - \mu_2 = \left(\hat{A}_2^\text{opt} - A_2^\text{opt}\right)\left(\psi_{20} + \psi_{21}\widehat{X}_2\right) + \left(A_2^\text{opt} - A_2\right)\psi_{21}'\left(X_2 - \widehat{X}_2\right). \label{eq:regret_estimation_error}\end{equation} If instead of the regret formulation, we compute the pseudo-outcome as described in the blip formulation (and as such do not need to estimate the optimal treatment), we will be left with the residual term $\gamma_2 - \widehat{\gamma}_2 = A_2\psi_{21}'\left(X_2 - \widehat{X}_2\right)$. While it is not possible, without access to $X_2$ directly, to guarantee that these two sources of error are completely eliminated, we are afforded some flexibility in how they are computed. In particular, if we estimate the blip parameters using the regression calibration correction, and assume that they have been correctly estimated, we could separately choose a covariate, $\W_2$, to use for estimating $\tilde{y}$. 

In the blip characterization $\W_2$ should be chosen in such a way as to minimize $X_2 - \W_2$. Noting that $\widehat{X}_2$ is chosen to be the (linear) estimator of $X_2$ which minimizes the mean squared error (MSE), this gives reasonable justification for selecting $\widehat{X}_2$. It is also worth noting that if $A_2 = 0$ the blip pseudo-outcome is exactly correct. As such, practitioners applying this method with this characterization may wish to consider a regression calibration conditional on $A_2 = 1$ (that is, estimate the BLUP only for those who received second stage treatment), which would minimize the MSE among linear estimators for only those patients who contribute to the biased pseudo-outcomes. 

The regret characterization warrants slightly more involved consideration. The second term in Equation (\ref{eq:regret_estimation_error}) has an impact dictated by $X_2 - \W_2$, as in the blip formulation. The first term, however, relies on a difference of indicator functions. If $\gamma_2 >> 0$ or $\gamma_2 << 0$, such that there is an unambiguous optimal treatment for the individual, then controlling $|\widehat{\gamma}_2 - \gamma_2|$ leads to $\widehat{A}_2^\text{opt} = A_2^\text{opt}$. In this situation, $\widehat{\gamma}_2$ near $\gamma_2$ simplifies to the condition that $\W_2$ is near $X_2$, and so we can once again rely on the justification that $\widehat{X}_2$ minimizes the MSE to motivate the selection of the regression calibration correction. If we have that $|\gamma_2| \leq \epsilon$ for a sufficiently small $\epsilon$, such that the optimal treatment is ambiguous, then it no longer suffices to have $\widehat{\gamma}_2$ near $\gamma_2$ (as even small perturbations between these quantities may lead to $\widehat{A}_2^\text{opt} \neq A_2^\text{opt}$). However, if we do have $\widehat{\gamma}_2$ near $\gamma_2$, then we can also make the claim that $|\widehat{\gamma}_2|$ is small, relatively speaking. The magnitude of the first term in Equation (\ref{eq:regret_estimation_error}) is given by $|\widehat{\gamma}_2|$, therefore, selecting an estimator to be near $\gamma_2$ will ensure that either (1) $\widehat{A}_2^\text{opt}$ is likely to be optimal in the event that there is a large treatment effect, or (2) that the magnitude of the error produced will be small when $\widehat{A}_2^\text{opt}$ is not optimal. This provides a heuristic rationale to use the regression calibration correction in order to estimate the pseudo-outcomes. In order to improve the MSE by conditioning, as was possible in the blip characterization, we would want to limit focus to patients for whom $A_2^\text{opt} \neq A_2$. However, $A_2^\text{opt}$ is not observable, and as such this is not a possible strategy.

There are obvious limitations to this justification. The first is that, in certain settings, it may be possible to derive an estimator which minimizes a loss function on the classification of optimal treatments. Further, restricting consideration to linear estimators of $X_2$ may be ill-advised. Finally, a metric other than MSE may be preferable to measure the distances in this setting. The first issue is a problem that is linked to optimal treatment recommendation (which we investigate briefly in Section \ref{sec:pseudo_correct_subsec}). The second concern extends beyond the estimation of pseudo-outcomes, and in such situations where linear estimators perform poorly, alternative corrections should be considered. There are extensions to regression calibration which provide higher order corrections which may be suitable.\cite{CarrollBook} Finally, where MSE is an inappropriate metric, practitioners of the methodology may be able to solve for an estimator which optimizes the desired metric instead. MSE is a generally applicable metric, which ought to serve well in a wide variety of scenarios. 

In order to perform dWOLS in an error-prone setting, we ultimately recommend computing the regression calibration estimates for all error-prone covariates, and then using these values in the treatment, treatment-free, and blip models, in addition to the estimation of the pseudo-outcomes. This procedure promises consistent parameter estimates under correct model specification, sample covariate balance using the weights, and a heuristic justification for acceptability of pseudo-outcomes. We caution any practitioner applying these methods to be mindful to their particular scenario, ensuring that the structures we have assumed are reasonable, and that our discussions remain valid for their use case.

\subsection{Confidence Intervals and Standard Errors}\label{confidence_intervals}
Regression calibration does not, in general, lend itself to the computation of closed-form variance estimators for the parameters of interest. There do exist derivations for asymptotic standard errors in generalized linear models, however bootstrapped confidence intervals tend to be the preferred solution.\cite{CarrollBook} In the case of dWOLS, there has been little theoretical development on closed-form variance estimators. They have been derived for the single-stage setting, where the authors caution that ``such variance estimates require careful calculation and coding, and so will likely not be practical for the typical analyst.'' indicating that bootstrap procedures seem to perform satisfactorily in their exploratory analyses.\cite{dwols} A modified bootstrap procedure, the \textit{m-out-of-n bootstrap}, was proposed for use in Q-learning to handle non-regularity concerns in the estimation of DTRs.\cite{mn_q_learning} The proposed adaptive procedure for selecting $m$ in Q-learning has been applied, with some success, to dWOLS.\cite{simoneau_mn_dwols} It seems that, where measurement error is a concern, a bootstrap procedure would presently be most suited for estimating confidence intervals for DTR parameters. 

We consider the m-out-of-n procedure, with an adaptive choice of $m$ to construct our intervals. We outline the fundamentals of the algorithm here, and advise the interested reader to consider the existing literature for a deeper exploration.\cite{mn_q_learning, simoneau_mn_dwols} The method performs a standard non-parametric bootstrap, where samples of size $m < n$ are drawn (with replacement), in place of the more conventional $n$. The theory dictates only that $m = o(n)$, and so in the finite sample case, we require a procedure for estimating $m$ from the data. We take \[m = n^\frac{1+\zeta(1-p)}{1+\zeta},\] where both $p$ and $\zeta$ are hyperparameters, selected from the data. The parameter $p$ is a measure of the non-regularity for the model in question, taking values in $[0,1]$. Of note, when $p = 0$, (where we have no regularity concerns), $m = n$ and this method is equivalent to the standard bootstrap. For a fixed value of $n$, $m \in [n^{1/(1+\zeta)}, n]$, and so $\zeta$ can be viewed as a parameter which controls the smallest acceptable re-sample size. 

We use an adaptive approach which estimates both $p$ and $\zeta$ from our data. Consider, for notational simplicity, a two-stage setting. Non-regularity concerns stem from patients for whom small perturbations in covariates lead to different optimal treatment decisions. As such, we take $\hat{p} = \widehat{P}(\widehat{\gamma}_2 = 0)$, which we estimate by considering the proportion of individuals who do not admit a unique optimal treatment decision at the second stage. That is, we construct confidence sets for the second stage blip, and count the proportion of individuals for whom this set contains $0$. To select $\hat{\zeta}$, we use a double-bootstrap procedure. 

We start by setting $\zeta$ to be a small value, and then draw $B_1$ samples of size $n$ from the initial data. Within each of these samples, we estimate $\hat{p}^{(b_1)}$ and the parameters of interest, $\widehat{\psi}^{(b_1)}$. We then conduct an m-out-of-n bootstrap procedure with $B_2$ iterations, using the current value of $\zeta$ and $\hat{p}^{(b_1)}$ to compute $\widehat{m}^{(b_1)}$. We use these $B_2$ resamples to form a confidence interval around the parameters of interest. This is repeated for each of the $B_1$ samples. We then check the nominal coverage probability, counting the proportion of the $B_1$ intervals which contain the initial estimate, and if this is at the desired level, we select the present value of $\zeta$ for $\hat{\zeta}$. Otherwise, we increment $\zeta$ and run the procedure again. The search space for $\zeta$ can be selected as necessary for the application, for instance, restricting the maximum considered value based on the smallest allowable re-sample size. Once $\hat{\zeta}$ and $\hat{p}$ are selected the bootstrap is performed with the estimated $\widehat{m}$.

\subsection{Future Treatments} \label{sec:pseudo_correct_subsec}
While we have focused on the identification of the optimal DTR, an important extension is to consider the implications that measurement error has on future treatment decisions. One consideration is to frame future treatment decisions as a prediction (or classification) problem. In such a framing, our goal is not to correctly estimate the causal parameters, but rather to correctly classify patients into their optimal treatment categories. As previously discussed, it is sometimes argued that measurement error corrections are unnecessary in a prediction setting, though this is not universally applicable.\cite{Pred_Butts} The complexity of DTRs suggest that it is worth considering the utility of error correction for predictions, and studying the effects of error being ignored.

In order to implement error correction for the assignment of future treatments we must consider what information is available when making those decisions. If we have measurements available for the entire population we wish to treat prior to making any treatment decision, we can apply regression calibration directly, and treat based on the imputed values. This setting is distinct from the situation where we are making treatment decisions one at a time, and consequently cannot pool the patients' information in order to apply regression calibration directly. In this setting we instead propose a \textit{pseudo-correction}, where the parameters required to adjust the covariates are made available, from the fitting stage, for use in the prediction stage. It is also conceivable, for instance due to cost constraints, that during the study we have error-prone covariates, while future decisions may be informed by the true covariates. Here, the prediction problem becomes one of predicting across domains.

\section{Simulation Studies}\label{sec:simulations}
We now demonstrate, via simulation, the potential impact of measurement error in the context of DTRs. We emphasize the issues that are present when conducting a naive analysis, and show the feasibility of regression calibration to largely correct for the errors in the analysis, as per our preceding discussions.

\subsection{Parameter Estimation}
We begin by demonstrating the bias present in blip parameter estimates resulting from a naive analysis, and the robustness of our proposed estimation procedures. First, we consider a simple one-stage setup, with $X\sim N(0,1)$, and assume that we observe two proxy measurements, given by $\W_1 \sim X + N(0, 0.25)$, and $\W_2 \sim X + t_{8}$. We assume that the treatment model is given by $P(A=1|\W_1=w) = H(1 - 0.5w + 1.5\exp(w - 1))$, and the outcome model is specified as $Y = X + \exp(X) + A(1+X) + \epsilon$ where $\epsilon \sim N(0,1)$, independent of all other variables. We are interested in estimating blip parameters $\psi_0 = 1$ and $\psi_1 = 1$.

In this setting, we consider four analyses, repeated with and without regression calibration, altering which components of our models are correctly specified. We fit models with (1) neither the treatment nor treatment-free models correctly specified, (2) only the treatment model correctly specified (where the treatment-free is taken to be linear), (3) only the treatment-free model correctly specified (where the treatment model is taken to be linear in the logistic scale), and (4) where both are correctly specified. In all scenarios we simulated $10000$ datasets of size $n=1000$. The results are summarized in Figure \ref{simulation_one_results}. 

\begin{figure}
	\centering
	\caption{Blip parameter estimates (true values $\psi_0 = \psi_1 = 1$ indicated by dashed lines) for 10000 simulated datasets (with $n=1000$), comparing a regression calibration corrected analysis to a naive analysis, when neither (Analysis 1), one of (Analyses 2 and 3), or both (Analysis 4) of the treatment and treatment-free models are correctly specified.} 
	\includegraphics[width=\textwidth]{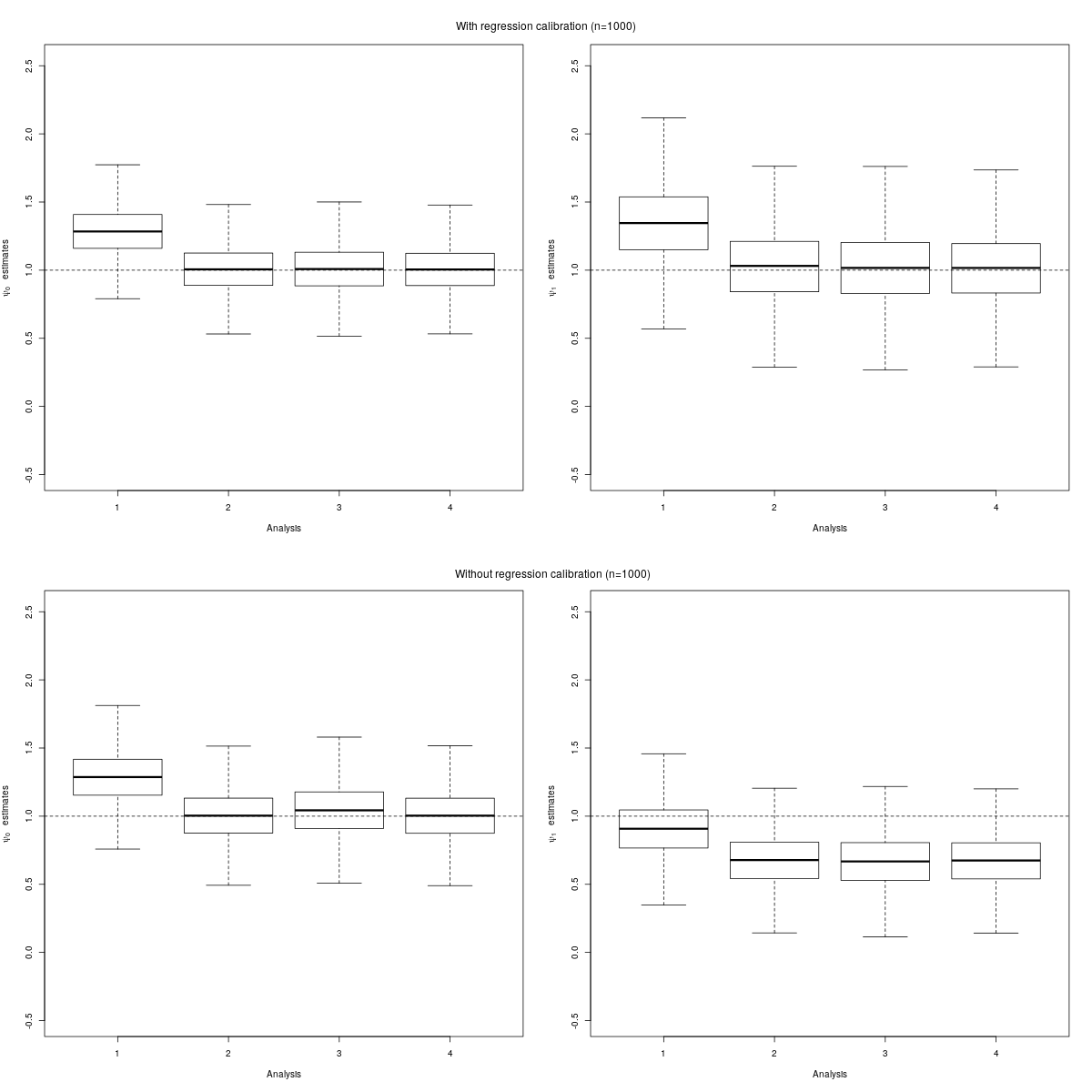}
	\label{simulation_one_results}
\end{figure}

When at least one model is correctly specified (analyses (2)-(4)), the naive estimators of $\psi_0$ perform well. In all four scenarios the naive results are biased for $\psi_1$. Regression calibration results in a clear improvement over the naive estimators in the results for $\psi_1$ across analyses (2)-(4), where the bias is largely removed. There is a clear, though dramatically reduced, bias in analysis (2), where the estimates rely on the correct specification of the treatment model alone. 

We further consider extending these analyses to a variety of two-stage DTR settings, adapted from the original dWOLS paper.\cite{dwols} In addition, we consider a scenario in which various error model combinations were used. We compare the estimation of the relevant parameters using the proposed correction to the parameter estimates obtained under an analysis using a weighted average of the available proxies. For all situations we consider basing treatment on only the first naive proxy, or on the mean of the available proxies. These results are summarized in the appendix, in Tables \ref{tb:multistage_one}-\ref{tb:multistage_five}. 

In general, we see that whether actual treatment decisions are based on a single error-prone covariate, or on the mean of multiple proxies, the correction methods are generally applicable. Across the majority of scenarios, the proposed corrections tend to greatly improve estimates compared to the naive analysis, and yield results which appear broadly consistent. The corrections work well across a variety of error mechanisms, where performance is only materially impacted when using a multiplicative gamma distribution to induce the error. These results confirm our comments regarding the importance of additive error models: the methods are somewhat resilient to these assumptions, but analysts should be careful when there is good reason to suspect a multiplicative model. When the treatment model is badly misspecified, we see notable degradation in the quality of the correction. These simulations suggest that careful consideration must be given to fitting the treatment model. 

\subsection{Coverage Probabilities}
Next, we consider three scenarios to test the applicability of the proposed bootstrap procedure. Due to the computational demands of the adaptive procedure, it is not feasible to conduct a full simulation study, adaptively selecting $\zeta$ for each experiment. Instead, we perform the double-bootstrap procedure once under each of the scenarios, and then consider the m-out-of-n bootstrap for values of $\zeta$ surrounding the selected one. In all three scenarios we take $X_1, X_2 \sim N(0,1)$, and observe two error prone proxies. Scenario 1 takes $\W_{11}, \W_{12} \sim X_1 + N(0,1)$, and $\W_{21}, \W_{22} \sim X_2 + N(0,1)$. Scenario 2 takes $\W_{11} \sim X_1 + N(0,1)$, $\W_{12} \sim X_1 + \text{Unif}(-1,1)$, $\W_{21} \sim X_2 + N(0,1)$, and $\W_{22} \sim X_2\cdot\text{Gamma}(1, 1)$. Finally, scenario 3 takes $\W_{11} \sim X_1 + \text{Unif}(-1, 1)$, $\W_{12} \sim X_1\cdot\text{Gamma}(1, 1)$, $\W_{21} \sim X_2 + N(0, 0.25)$, and $\W_{22} \sim X_2 + \text{Unif}(-1, 1)$. For all three scenarios, we take $P(A_j = 1|\W_{j1}=w) = H(w)$. The outcome for scenarios 1 and 2 is given by $Y = X_1 + X_2 + A_1(1 + X_1) + A_2(1 + X_2) + \epsilon$ where $\epsilon \sim N(0,1)$ independent of everything else. For scenario 3, we introduce an additional binary covariate, $Z_2$, with $P(Z_2 = 1) = 0.5$. We then take $Y = X_1 + X_2 + A_1(1 + X_1) + A_2(1 + X_2 - Z_2 - Z_2X_2) + \epsilon$ where, again, $\epsilon \sim N(0,1)$. Note that, if $Z_2 = 1$ then $\gamma_2 = 0$, meaning that the optimal treatment is not well-defined. 

In the first two scenarios we estimate $\hat{\zeta} = 0.05$, while in the third scenario $\hat{\zeta} = 0.075$. For all scenarios we consider forming bootstrap confidence intervals using (1) a traditional n-out-of-n bootstrap, (2) an m-out-of-n bootstrap where $\zeta=0.05$ is used in the adaptive procedure, and (3) an m-out-of-n bootstrap where $\zeta = 0.10$ is used in the adaptive procedure. For the third scenario, we also include an m-out-of-n bootstrap where $\zeta=0.075$ is used. The coverage probabilities are contained in Table \ref{coverage_sim_results}. We see that the standard bootstrap procedure attained the nominal coverage in all settings. Taking the selected $\widehat{\zeta}$ met the nominal coverage levels in the second scenario, and was slightly conservative for the first and third scenarios, where taking $\zeta = 0.10$, we obtained mostly conservative intervals. In the third scenario, all procedures tended to produce conservative results.

\subsection{Future Treatment Predictions}
To investigate future treatment assignment, we consider a two-stage DTR where $X_1 \sim N(0,1)$ and $X_2 \sim N(A_1, 1)$ are the true covariates, with replicate observations $\W_{11} \sim X_1 + t_{10}$, $\W_{12} \sim X_1 + N(0,1)$ and $\W_{21}, \W_{22} \sim X_2 + N(0, 0.25)$. The outcome is given by $Y = X_1 - (A_1^\text{opt} - A_1)(1 - X_1) - (A_2^\text{opt} - A_2)(3 - 2X_2) + \epsilon$, with $\epsilon \sim N(0,2)$ independent of all other parameters. The treatment models take the form $P(A_i = 1|\W_{i1} = w) = H(1 - w)$. 

We partition these analyses into three settings based on the information available at the time treatment decisions (or predictions) are made. Namely, where we have access to (1) error-prone measurements for only one patient at a time, (2) error-prone measurements for all patients at once, and (3) the true covariate values for prediction. For these scenarios we consider the performance of the naive model compared to the corrected model. For the first setting, direct regression calibration is not possible. Instead, we conduct the pseudo-correction described in Section \ref{sec:pseudo_correct_subsec}, to produce a corrected estimate. We provide the results for this correction applied when we measure both proxies and when we only measure a single proxy. In the second setting, we do not fit the naive model (as it is equivalent to scenario (1)). Each analysis above is run with $n=1000$ individuals during the fitting stage, and the treatment assignment is run for $5000$ individuals. We repeat the set of simulations $10000$ times. 

\begin{figure}
	\centering
	\caption{Proportion of optimally treated individuals for 10000 simulated datasets ($n=1000$), comparing the predictive power of naive versus corrected analysis when we (1) make future treatment decisions one-at-a-time, (2) have recorded information in order to pool patient information, or (3) have the true measurements available to inform the treatment decision. When single treatment decisions are made, we use a pseudo-correction based on all error-prone proxies (corrected*) or on a single error-prone proxy (corrected**), using the same parameter estimates for the pseudo-correction in all three cases. Ranges of optimally treated individuals are shown at stage one and stage two, separately.}
	\includegraphics[width=\textwidth]{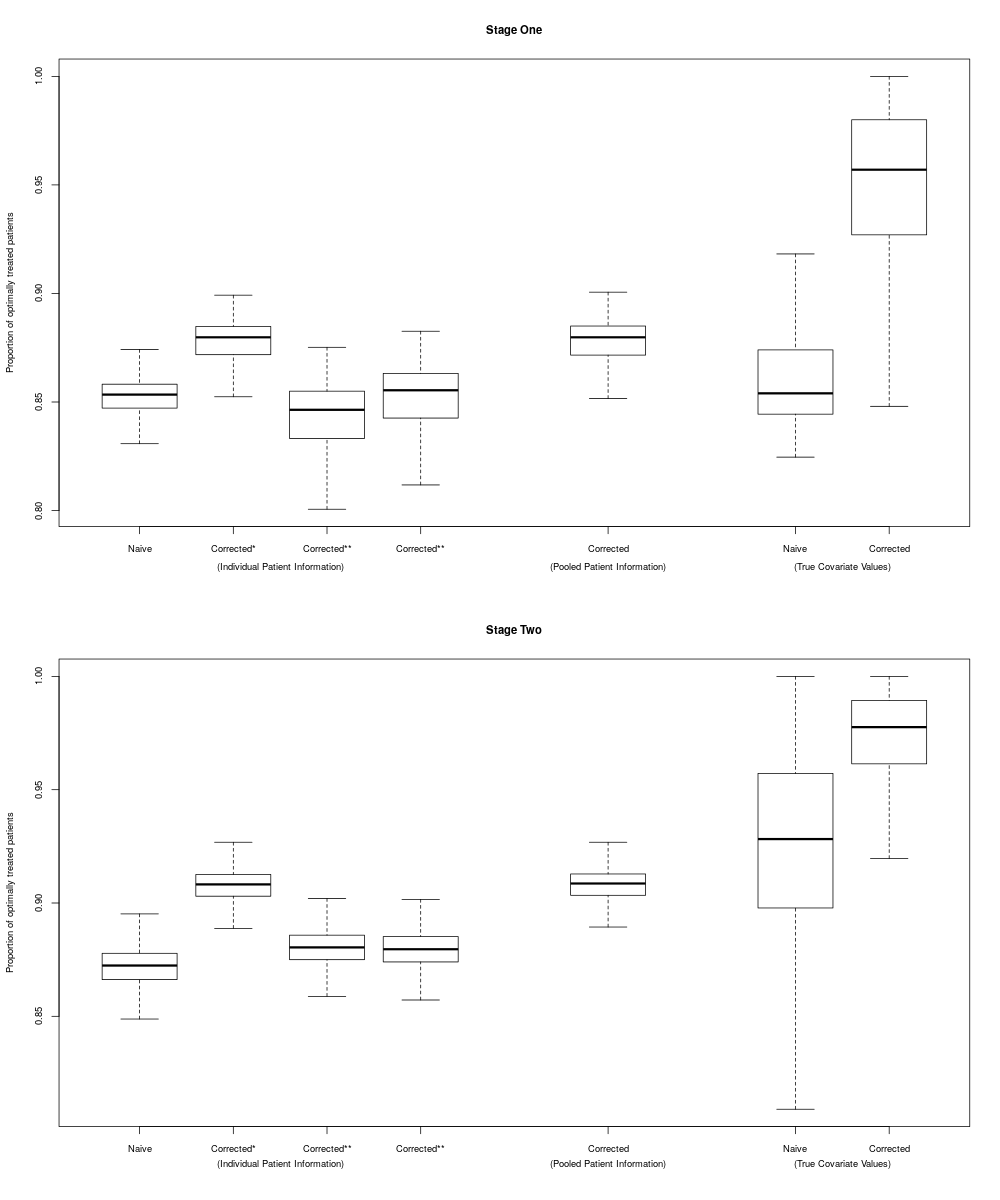}
	\label{simulation_three_results}
\end{figure}

The complete results are provided in Figure \ref{simulation_three_results}, where, on average, the corrected methods perform better than the naive methods in terms of accuracy. The results suggest that, in the worst case scenario, framed as a problem of prediction, the naive and corrected methods perform comparably. However, there are dramatic gains in terms of optimality of treatment in the event that additional information is available when assigning future treatments. At stage one, the pseudo-correction performs favorably using only the second proxy, as compared to only the first proxy, as a result of its lower variance and higher weight while fitting the estimator. This suggests that if the pseudo-correction is to be used, and only one proxy will be made available, the proxy with the lowest variance is preferred.

\section{STAR*D study}\label{STARD}
We now illustrate the proposed correction methods through application to data from the Sequenced Treatment Alternatives to Relieve Depression (STAR*D) study. The STAR*D study was a multistage randomized controlled trial, comparing different treatment regimes for patients with major depressive disorder.\cite{Fava2003, Rush2004} The study was split into four phases (with phase two further subdivided into two sub-phases) where, at each phase, different treatment options were available to patients based on preference and progression through the study. The severity of depression was measured through the Quick Inventory of Depressive Symptomatology (QIDS) score, where assessment was conducted during each phase both by the patient (denoted QIDS-S) and by a clinician (denoted QIDS-C). At the end of each study phase, patients who had a clinician assessed QIDS score less than or equal to $5$ were considered to have entered remission, and were subsequently removed from the study. At phase 1, all patients were prescribed citalopram. At the end of phase 1, those who did not enter remission entered the second stage where seven treatment options were available: this phase was characterized by `switching' from citalopram to one of four other treatments options, or `augmenting' treatment by receiving citalopram alongside one of three new treatments. Those who had still not entered remission entered a third (and possibly fourth) phase, where treatment was again switched our augmented with a variety of possible options. Full details of the study, and of the treatment options, are described elsewhere.\cite{Rush2004} 

The first and fourth phases of the trial are typically ignored in DTR analysis of these data. We focus on phases two (merging both sub-phases) and three, which we refer to as stage one and stage two, respectively. Previous analyses \cite{Chakraborty2013} specified QIDS-C as an outcome, and dichotomized treatments differentiating those which contain a selective serotonin reuptake inhibitor (SSRI) and those which do not. These analyses model QIDS-C as a continuous covariate and consider three tailoring variables: QIDS-C measured at the start of each level (given by $Q_j$ for stage $j$), the change in QIDS-C divided by the elapsed time over the previous level (referred to as QIDS slope, denoted $S_j$ for stage $j$), and patient preference (denoted $P_j$ for stage $j$), a binary indicator specifying whether the patient desired to switch treatment regimes ($P_j = 1$) or augment ($P_j = 0$). Treatment is coded as $A_j = 1$ if the stage $j$ treatment includes an SSRI, and a $0$ otherwise. Our analysis considers 283 patients, who had all stage one and two covariates measured. The outcome was taken to be $Y = -\frac{1}{2}\left(\text{QIDS-C}_{1} + \text{QIDS-C}_{2}\right)$, where $\text{QIDS-C}_{j}$ is the clinician rated QIDS score at the end of stage $j$. 

Existing analyses of these data make the implicit assumption that clinician scores are error-free measurements. However, the inclusion of self-assessed measures (QIDS-S), offers a feasible mechanism for exploring measurement error in this study. If we postulate that there exists a true underlying symptom score for every patient, then we might propose that both the self-assessed and the clinician scores are surrogate measures for this truth, permitting regression calibration. We note that our analysis continues to use QIDS-C as the outcome variable, to remain comparable with previous literature. 

\subsection{Model Fitting and Comparison}
We consider fitting the model using only the clinician ratings, only the self-reports, or using the correction where they are considered to be error-prone proxies. Following previous analyses of the data, we pose a first stage treatment model using only first stage preference ($P_1$) and a second stage treatment model using only second stage preference ($P_2$). For the first stage, the treatment-free and blip models are linear in preference ($P_1$), slope ($S_1$), and initial QIDS score ($Q_1$). At the second stage, the treatment-free model is linear in preference ($P_2$), slope ($S_2$), starting value ($Q_2$), as well as stage one treatment ($A_1$). The blip model used only slope ($S_2$) and starting value ($Q_2$). For each of the settings we conducted an m-out-of-n bootstrap, choosing $m$ using the outlined adaptive procedure. Table \ref{stard_results} contains the results for parameters estimates and $95\%$ confidence intervals.

Previous analyses have found that the only significant treatment effect was the interaction between stage one treatment and preference ($A_1P_1$), \cite{mn_q_learning} a result that is replicated on our subset of the data when using only clinician scores. If instead we assume that the self-reported scores represent the true values, we find a significant treatment effect at stage two, with the interaction between treatment and slope ($A_2S_2$). However, if we perform our correction, neither of these effects remains significant, and we lack evidence for any significant treatment effects. This may be due to increased uncertainty from the error, but it nevertheless suggests further consideration is required.

\section{Discussion}
Dynamic treatment regimes provide a powerful framework for characterizing treatment pathways. The theory surrounding optimal DTR estimation is well-developed, with numerous methods available. Measurement error is a pervasive issue in many scenarios and has consequently received extensive attention. However, there has been no substantive work which investigates DTRs in an error-prone setting. The errors arising due to measurement error in the DTR framework are somewhat unique, as the treatment-free, treatment, and blip models are all affected by measurement error separately. The treatment model no longer depends on the true covariate values, as these are unobserved, adding extra complexity to the modeling procedure. Additionally, fitting a DTR often requires patient weights designed to induce covariate balance between a patient's treatment and their covariates. This balance may not be guaranteed in the presence of measurement error. 

Despite these additional considerations, since DTRs can be effectively estimated through regression methods, and since measurement error has been well-studied in regression frameworks, extant error correction methods offer a feasible solution. We investigated the use of regression calibration to correct for covariate measurement error, in the tailoring covariates, in DTRs with a continuous outcome. We have demonstrated that the application of regression calibration within the dWOLS analysis framework is an effective technique for substantially reducing the bias present in an analysis that does not attempt to account for measurement error. Further, the estimates tended to exhibit desirable behavior across a wide variety of settings, largely preserving the doubly robust property of dWOLS.

The multistage setting poses additional considerations due to the need to estimate the sequence of pseudo-outcomes. Even when the blip parameters are correctly estimated there is residual bias when constructing the pseudo-outcome from any error-prone proxy. We argue that using the regression calibration correction to form the pseudo-outcomes can be justified, using MSE as a metric across the class of linear estimators for the true covariates. While this argument is largely heuristic, the simulations which were conducted tend to confirm that this procedure often suffices to correct the parameter estimates. There is opportunity for further theoretical development of these concepts, such as considering what metrics are best-suited for the assessment of pseduo-outcome estimates, both when covariates are subject to error and when they are not.

The m-out-of-n bootstrap with an adaptive choice for $m$, as proposed to handle non-regularity in the case of error-free DTRs, seems promising as a mechanism for producing nominal, or at worst conservative, confidence intervals while using the proposed correction. This is in line with the related theory in both the measurement error and DTR literature, where the m-out-of-n bootstrap has shown to be effective for dWOLS, and bootstrap methods are broadly useful with regression calibration. While there is substantial room for the development of rigorous confidence intervals using dWOLS, both in the error-prone and error-free settings, the recommended m-out-of-n procedure seems to be well-suited for the task.

While it may seem plausible to view the process of estimating a DTR to form a treatment rule for future patients as a question of prediction, removing the need for error correction methods, we have demonstrated that the naive estimates are more highly variable and tend to provide lower rates of optimal treatment as compared to a corrected analysis. Further, when there is the possibility of observing error-free patient information when making future treatment decisions, the gain from using regression calibration in the study is substantial, in terms of the proportion of optimally treated individuals. 

Finally, we demonstrated our methods through analysis of the STAR*D study data, where additional error-prone proxies were available due to the self-reported QIDS scores, which have typically been ignored in prior analyses of this dataset. Accounting for errors alters the estimated optimal treatment rules, which may be attributable to the increased uncertainty that measurement error induces, though it suggests that the true optimal decision rules likely require further investigation to be correctly identified.

The resolution of measurement error for DTRs is a complex challenge. The approaches proposed in this work rely on a number of assumptions that may not hold in practice, such as the availability of unbiased instrumental data or the applicability of the error models explored. Nevertheless, we have demonstrated both the potential impact of measurement error on so-called naive analyses, and the ease with which considerable gains can be made through the application of comparatively straightforward analytical techniques. More complex methods (both from the DTR and measurement error literatures) may elicit further improvements, and we anticipate pursuing these in future work. Beyond this, errors in the treatment and outcome variables must also be considered, both of which may be investigated as natural extensions to the results presented here.

\section*{Acknowledgments}
This work was funded by a Natural Sciences and Engineering Research Council (NSERC) Discovery Grant. The data for the STAR*D study are available from the National Institute of Mental Health (NIMH). Restrictions apply to the availability of these data, which were used under license for this study. Data are available through the NIMH Data Archive (NIMH NDA ID: 2148). We are grateful for the feedback from the reviewers and editors on an early draft of this manuscript. The provided comments contributed important developments to the methodologies presented.

\section*{Data Availability Statement}
All of the R code used to run the simulation studies and sensitivity analyses is freely available online at: https://github.com/DylanSpicker/measurement-error-DTRs. The data for the STAR*D study are available from the National Institute of Mental Health (NIMH). Restrictions apply to the availability of these data, which were used under license for this study. Data are available through the NIMH Data Archive (NIMH NDA ID: 2148).




\bibliographystyle{unsrt}
\bibliography{references}

\pagebreak

\begin{table}[h]
    \centering
    \caption{Results for the coverage probabilities derived from bootstrap procedures across three scenarios. Intervals formed using an n-out-of-n ($\text{nn}$) or an m-out-of-n bootstrap, based on the adaptive procedure with $\zeta$ ($\text{mn}_\zeta$). Each set was constructed using $2000$ bootstrap replicates, and the experiment was repeated $500$ times. Bolded values indicate those which deviate significantly from the nominal coverage of $0.95$. Intervals are shown for treatment at both stages ($A_j$), as well as treatment interactions with the error-prone covariates ($X_j$) and the error-free covariate $Z_2$. Only scenario 3 used the error-free covariate.}
    \begin{tabular}{rrrrrrrrrrrrr}
    \toprule & \multicolumn{3}{c}{Scenario One} & & \multicolumn{3}{c}{Scenario Two} & & \multicolumn{4}{c}{Scenario 3} \\
    & $\text{nn}$ & $\text{mn}_{.05}$ & $\text{mn}_{.10}$ & & $\text{nn}$ & $\text{mn}_{.05}$ & $\text{mn}_{.10}$ & & $\text{nn}$ & $\text{mn}_{.075}$ & $\text{mn}_{.05}$ & $\text{mn}_{.10}$ \\
    \toprule $A_1$ & 0.937 & 0.958 & \textbf{0.970} & & 0.950 & 0.962 & \textbf{0.972} & & \textbf{0.97} & \textbf{0.98} & \textbf{0.98} & \textbf{0.98} \\
    $A_1X_1$ & 0.964 & \textbf{0.970} & \textbf{0.979} & & 0.950 & 0.958 & 0.962 & & \textbf{0.98} & \textbf{0.99} & \textbf{0.99} & \textbf{0.99} \\
    $A_2$ & 0.961 & \textbf{0.970} & \textbf{0.979} & & 0.950 & 0.964 & \textbf{0.980} & & \textbf{0.97} & \textbf{0.99} & \textbf{0.99} & \textbf{1.00} \\
    $A_2X_2$ & 0.940 & 0.961 & \textbf{0.973} & & 0.952 & 0.956 & 0.960 & & 0.96 & \textbf{0.99} & \textbf{0.98} & \textbf{0.99} \\
    $A_2Z_2$ & -- & -- & -- & & -- & --  & -- & &  \textbf{0.97} & \textbf{0.98} & \textbf{0.98} & \textbf{0.98} \\
    $A_2X_2Z_2$ & -- & -- & -- & & -- & --  & -- & & 0.97 & \textbf{0.99} & \textbf{0.98} & \textbf{0.99} \\
    \hline
    \end{tabular}
    \label{coverage_sim_results}
\end{table}

\begin{table}[ht]
	\caption{Results for the two-stage blip coefficient estimates comparing an analysis employing the regression calibration correction to naive analysis using only the clinician or self-reported data. Confidence intervals are computed based on 2000 m-out-of-n bootstrap replicates, where $m$ was chosen based on the described adaptive procedure. Bolded values indicate treatment effects which are significant at a $95\%$ level. $A_j$ refers to the treatment indicator ($1$ for those with an SSRI, $0$ otherwise), $P_j$ refers to patient preference to switch ($1$ with a preference to switch, $0$ with a preference to augment), $Q_j$ refers to the starting QIDS score at stage $j$, and $S_j$ the slope of the QIDS score over the $j$-th phase.}
    \centering
    \begin{tabular}{lcccccc}
        \toprule & \multicolumn{2}{c}{Error Corrected} & \multicolumn{2}{c}{Clinician Score} & \multicolumn{2}{c}{Self-Reported} \\ 
        \cmidrule(lr){2-3} \cmidrule(lr){4-5} \cmidrule(lr){6-7}  
        Parameter & Estimate & $95\%$ CI & Estimate & $95\%$ CI & Estimate & $95\%$ CI \\
        \toprule
        \multicolumn{7}{l}{\rule{0pt}{3ex} Stage One} \\
        $A_1$ & -0.75 & (-10.035, 7.934) & -0.48 & (-6.279, 5.486) & 1.35 & (-3.784, 6.052) \\
        $A_1P_1$ & 2.72 & (-0.186, 5.819) & \textbf{2.99} & \textbf{(0.898, 5.426)} & 2.76 & (-0.199, 5.826)  \\
        $A_1Q_1$ & 0.06 & (-0.569, 0.701)& 0.07 & (-0.346, 0.456) & -0.09 & (-0.409, 0.253)  \\
        $A_1S_1$ & -1.54 & (-6.895, 2.246) & -1.04 & (-3.768, 1.075) & -0.55 & (-2.312, 0.887) \\
        \multicolumn{7}{l}{\rule{0pt}{3ex} Stage Two}\\
        $A_2$  & -0.31 & (-7.045, 6.954) & 1.19 & (-2.88, 5.521) & -0.04 & (-4.262, 4.078) \\
        $A_2Q_2$ & 0.09 & (-0.479, 0.632) & -0.02 & (-0.371, 0.3) & 0.08 & (-0.219, 0.383) \\
        $A_2S_2$ & 1.82 & (-2.789, 4.833) & 0.94 & (-0.82, 2.696) & \textbf{2.74} & \textbf{(0.297, 5.137)} \\
       \hline
    \end{tabular}
	\label{stard_results}
\end{table}

\pagebreak

\appendix
\section{Regression Calibration: Details} \label{apdx:RC_details}
We first consider defining the optimal weights $\delta_j$. In our investigation, we have paid most attention to three separate sets of weights: (1) using $\delta_j = \frac{1}{k}$ for all $k$ proxy measurements, (2) viewing $\W$ as an estimator for $X$, and minimizing the variance of that estimator, which gives $\delta_j = \Tr\left(M_j\right)^{-1}\left[\sum_{l=1}^k \Tr\left(M_l\right)^{-1}\right]^{-1}$, or (3) treating $\delta_j$ as parameters in the BLUP and solving, which gives the form $\delta_j = \Tr\left(\beta'\beta M_j\right)^{-1}\left[\sum_{l=1}^k \Tr\left(\beta'\beta M_l\right)^{-1}\right]^{-1}$, which we can solve numerically. Here $M_j$ refers to the $j$-th matrix in the covariance term. Solving for the BLUP gives us Equation (\ref{eq:simplified_form}). To implement each of these, we advocate for simple plug-in estimators for each of the corresponding quantities, most of which can be readily derived through ANOVA-style calculations. Unless otherwise stated, we use 
\begin{align*} 
    \overline{X^*_{\cdot j}} &\equiv \frac{1}{n}\sum_{i=1}^n X^*_{ij}\\
    \overline{X_i^{*(j)}} &\equiv \frac{1}{k-1}\sum_{\substack{l=1 \\ l\neq j}}^k X^*_{il} \\
    \widehat{\Sigma}_{X_j^*} &= \frac{1}{n-1}\sum_{i=1}^n \left(X^*_{ij} - \overline{X^*_{\cdot j}}\right)\left(X^*_{ij} - \overline{X^*_{\cdot j}}\right) \\
    \widehat{M} &= \frac{k-1}{kn}\sum_{i=1}^n\sum_{j=1}^k \left(X^*_{ij} - \overline{X^{*(j)}_{i}}\right)\left(X^*_{ij} - \overline{X^{*(j)}_{i}}\right)' \\
    \widehat{\Sigma}_{XX}^{(1)} &= \frac{1}{k}\left\{\sum_{j=1}^k \widehat{\Sigma}_{X_j^*} - \widehat{M}\right\} \\
    \widehat{M}_j &= \widehat{\Sigma}_{X_j^*} - \widehat{\Sigma}_{XX}^{(1)}\\
    \widehat{\mu}_X = \widehat{\mu}_{X^*} &= \frac{1}{n}\sum_{i=1}^n X_i^* = \frac{1}{n}\sum_{i=1}^n\sum_{j=1}^k \delta_jX^*_{ij} \\
    \widehat{\mu}_Z &= \frac{1}{n}\sum_{i=1}^n Z_i \\
    \widehat{\Sigma}_{ZZ} &= \frac{1}{n-1}\sum_{i=1}^n \left(Z_i - \overline{Z}\right)\left(Z_i - \overline{Z}\right)' \\
    \widehat{\Sigma}_{XZ} = \widehat{\Sigma}_{X^*Z} &= \frac{1}{n-1}\sum_{i=1}^n\left(X^*_{i} - \overline{X^*}\right)(Z_i - \overline{Z})' \\
    \widehat{\Sigma}_{X^*} &= \frac{1}{n-1}\sum_{i=1}^n \left(X^*_i - \overline{X^*}\right)\left(X^*_i - \overline{X^*}\right)' \\
    \widehat{\Sigma}_{XX}^{(2)} &= \widehat{\Sigma}_{X^*} - \sum_{j=1}^k\delta_j^2\widehat{M}_j.
\end{align*}

\section{Proofs of Result 1}
\label{Result_1_Proof}
\begin{proof}[Proof of Result \ref{covariate_balance}]
    Consider an arbitrary logistic regression, explaining $A$ with respect so some covariate $T$. We will assume, WLOG, that $T$ is univariate to suppress vector notation, though the same argument holds for vector-valued covariates. Assume that the model is fit with a sample of $(A_1, t_1), \hdots, (A_n, t_n)$. Computing a maximum likelihood estimate for $P(A=1|T)$ results in solving 
    \[
        \frac{\partial\ell}{\partial\alpha_0} = \sum_{i=1}^n \frac{A_i - (1 - A_i)\exp(\alpha_0 + \alpha_1t_{i})}{1 + \exp(\alpha_0 + \alpha_1t_{i})}
    \] 
    and
    \[
        \frac{\partial\ell}{\partial\alpha_1} = \sum_{i=1}^n t_{i}\frac{A_i - (1 - A_i)\exp(\alpha_0 + \alpha_1t_{i})}{1 + \exp(\alpha_0 + \alpha_1t_{i})}
    \] 
    equal to $0$. Solving these simultaneously results in 
    \[
        \sum_{i=1}^n \frac{A_i}{1 + \exp(\alpha_0 + \alpha_1t_{i})} = \sum_{i=1}^n  \frac{(1 - A_i)\exp(\alpha_0 + \alpha_1t_{i})}{1 + \exp(\alpha_0 + \alpha_1t_{i})}
     \] 
     and
     \[
        \sum_{i=1}^n \frac{t_{i}A_i}{1 + \exp(\alpha_0 + \alpha_1t_{i})} = \sum_{i=1}^n \frac{t_{i}(1 - A_i)\exp(\alpha_0 + \alpha_1t_{i})}{1 + \exp(\alpha_0 + \alpha_1t_{i})}.
    \]
    
    The above expressions simplify to:
    \begin{align*}
    	\sum_{i=1}^n A_i(1-\widehat{P}(A=1|T=t_i)) &= \sum_{i=1}^n (1-A_i)\widehat{P}(A=1|T=t_i) \\
        \sum_{i=1}^n t_{i}A_i(1-\widehat{P}(A=1|T=t_i)) &= \sum_{i=1}^n t_{i}(1-A_i)\widehat{P}(A=1|T=t_i) 
    \end{align*}

    This of course still holds for $T=\widehat{X}$, since the exact form of $T$ did not factor into the expression. Then, if the weights are defined to be $|A - \widehat{P}(A=1|\widehat{X})|$, we can see that the ratio of the above two expressions gives us the required form for sample covariate balance. 
\end{proof}

\section{Multistage Simulation Results}
To investigate the procedure in the multistage scenario, we consider a variety of related settings formed by varying different aspects of the model. We take $X_1 \sim N(0,1)$, with $\W_{11} \sim g_1(X_1)$ and $\W_{12} \sim g_2(X_1)$, for error models $g_1, g_2$. We assume that $P(A_1 = 1|\W_1=\w_1) = h_1(\w_1; \alpha_{10}, \alpha_{11})$, with a treatment model $h_1$ and parameters $\alpha_{10}, \alpha_{11}$. We also take $X_2 \sim N(A_1, 1)$, with $\W_{21} = g_1(X_2)$ and $\W_{22} = g_2(X_2)$, and $P(A_2 = 1|\W_2=\w_2) = h_2(\w_2; \alpha_{20}, \alpha_{21})$. The outcome is given by $Y = f(X_1) + (A_1^\text{opt} - A_1)\left(1 + \psi_{11}X_1\right) + (A_2^\text{opt} - A_2)\left(1 + \psi_{21}X_2\right) + \epsilon$, with $\epsilon \sim N(0,1)$, where $f(X_1)$ is the treatment-free model. We consider five scenarios by altering the above parameters. \begin{enumerate}
    \item Considers 10 combinations of $(\alpha_{10}, \alpha_{20})$, values taken from $\{-2, -1, 0, 1, 2\}$, holding the treatment-free model as linear, both treatment models as linear, the error models as classical additive with $N(0, 0.25)$ distribution, $\psi_{11} = \psi_{21} = 1$.
    \item Considers 10 combinations of $(\psi_{11}, \psi_{21})$, values taken from $\{-1, -0.1, 0, 0.1, 1\}$, holding $\alpha_{10} = \alpha_{20} = 0$, the treatment-free model as linear, both treatment models as linear, the error models as classical additive with $N(0, 0.25)$ distribution. 
    \item Considers 5 scenarios for various forms of the treatment-free model, taking $f(X_1) = X_1$ (linear), $f(X_1) = X_1 + X_1^2$ (quadratic), $f(X_1) = X_1 + X_2^2 - X_1^3$ (cubic), $f(X_1) = \exp(X_1) - X_1^3$ (exponential), or $\exp(X_1)I(X_1 >= -0.5)$ (complex). We hold both treatment models to be linear, $\alpha_{10} = \alpha_{20} = 0$, $\psi_{11} = \psi_{21} = 1$, and the error the error models as classical additive with $N(0, 0.25)$ distribution. 
    \item Considers 10 scenarios where the treatment models are taken to be one of $h_j(\w_j) = \alpha_{j0} + \alpha_{j1}\w_j$ (linear), $h_j(\w_j) = \alpha_{j0} + \alpha_{j1}\w_j + (\w_j)^2$ (quadratic), $h_j(\w_j) = \alpha_{j0} + \alpha_{j1}\w_j + \exp(\w_j)$ (exponential), and $h_j(\w_j) = \alpha_{j0} + \alpha_{j1}\w_j + (\w_j)^2 + \exp(\w_j)$ (mixed). We hold the treatment-free model to be linear, $\alpha_{10} = \alpha_{20} = 0$, $\psi_{11} = \psi_{21} = 1$, and the error the error models as classical additive with $N(0, 0.25)$ distribution.
    \item Considers 10 scenarios for various error models, taking $g_j(X_l) = X_l + N(0, 0.25)$ (normal), $g_j(X_l) = X_l + t_{10}$ (approximately normal), $g_j(X_l) = X_l\cdot\text{Gamma}(1, 1)$ (gamma), or $g_j(X_l) = X_l\cdot\text{Unif}(0.5, 1.5)$ (uniform). We hold the treatment-free model to be linear, both treatment models to be linear, $\alpha_{10} = \alpha_{20} = 0$, and $\psi_{11} = \psi_{21} = 1$.
\end{enumerate}

All analyses are conducted where $(\W_1, \W_2)$ is taken to be $(\W_{11}, \W_{21}), (\overline{\W_1}, \overline{\W_2}), (\W_{11}, \overline{\W_2}), (\overline{\W_1}, \W_{21})$ (that is treatment either depends on the first naive proxy, or on the mean of the two proxies). We take $n=10000$ and repeat each scenario $1000$ times. The results for a corrected analysis and a naive analysis are included in Tables \ref{tb:multistage_one}-\ref{tb:multistage_five}.

\begin{table}[ht]
    \caption{Median parameter estimates investigating the impact of treatment probabilities in a multistage DTR, by varying ($\alpha_{10}, \alpha_{20}$) as indicated. Blip parameter estimates are compared for $n=10000$ individuals, using the corrected method compared to a naive analysis. The top set of rows of the table use the first error-prone proxy at both stages, the second set of rows use the mean of proxies at both stages, the third set of rows use the mean at the first stage and the first error-prone proxy at the second, and the final set of rows use the first error-prone proxy at the first stage and the mean at the second. Bold values indicate parameters for which the $95\%$ percentile-based interval across the $1000$ simulation replicates did \textbf{not} cover the true parameter value.}
    \centering
    \begin{tabular}{lcccccccccc}
        \toprule & \multicolumn{4}{c}{Regression Calibration} & & \multicolumn{4}{c}{Naive}\\
        $(\alpha_{10}, \alpha_{20})$ & $A_{1}$ & $A_1X_{1}$ & $A_{2}$ & $A_2X_{2}$ & & $A_{1}$ & $A_1X_{1}$ & $A_{2}$ & $A_2X_{2}$ \\
        \toprule 
        (-2, -2) & 1.0101 & 0.9974 & 1.0035 & 0.9963 &  & 1.0098 & \textbf{0.8874} & 1.0163 & \textbf{0.902} \\
        (-1, -1) & 1.0103 & 0.9983 & 1.0038 & 0.9961 &  & 1.0102 & \textbf{0.8879} & 1.0304 & \textbf{0.9064} \\
        (0, 0) & 1.0106 & 0.9961 & 1.0003 & 0.9995 &  & 1.0108 & \textbf{0.8859} & 1.045 & \textbf{0.9101} \\
        (1, 1) & 1.0086 & 0.996 & 0.998 & 0.9994 &  & 1.0086 & \textbf{0.8857} & \textbf{1.0622} & \textbf{0.909} \\
        (2, 2) & 1.0082 & 0.9983 & 0.9975 & 1.0027 &  & 1.0083 & \textbf{0.8881} & 1.0789 & \textbf{0.9077} \\
        (-2, 0) & 1.01 & 0.9973 & 1.0009 & 0.9987 &  & 1.01 & \textbf{0.8872} & 1.0157 & \textbf{0.8993} \\
        (-1, 1) & 1.0109 & 0.9977 & 0.999 & 0.9986 &  & 1.0109 & \textbf{0.8868} & 1.0267 & \textbf{0.9044} \\
        (0, 2) & 1.011 & 0.9963 & 0.9963 & 0.9987 &  & 1.0105 & \textbf{0.8859} & 1.0413 & \textbf{0.9079} \\
        (1, -2) & 1.0087 & 0.9971 & 1.0028 & 0.9993 &  & 1.0084 & \textbf{0.886} & 1.0704 & \textbf{0.904} \\
        (2, -1) & 1.0084 & 0.9997 & 1.0036 & 0.9975 &  & 1.0079 & \textbf{0.8886} & \textbf{1.0891} & \textbf{0.8982} \\
        \hline
        (-2, -2) & 1.008 & 0.9992 & 1.0034 & 0.9965 &  & 1.0087 & \textbf{0.8887} & 1.0165 & \textbf{0.902} \\
        (-1, -1) & 1.0112 & 0.9988 & 1.0048 & 0.998 &  & 1.0107 & \textbf{0.8877} & 1.031 & \textbf{0.9068} \\
        (0, 0) & 1.0111 & 0.997 & 1.0013 & 0.9967 &  & 1.0111 & \textbf{0.8872} & 1.0454 & \textbf{0.9074} \\
        (1, 1) & 1.0093 & 0.9963 & 0.9975 & 0.9985 &  & 1.009 & \textbf{0.8865} & \textbf{1.0621} & \textbf{0.9083} \\
        (2, 2) & 1.0101 & 0.9993 & 0.9977 & 1.0007 &  & 1.0104 & \textbf{0.8884} & 1.0795 & \textbf{0.9054} \\
        (-2, 0) & 1.0108 & 0.9994 & 1.0011 & 0.9985 &  & 1.0108 & \textbf{0.8884} & 1.0159 & \textbf{0.8991} \\
        (-1, 1) & 1.011 & 0.9982 & 0.9979 & 0.9988 &  & 1.0113 & \textbf{0.8875} & 1.0257 & \textbf{0.9039} \\
        (0, 2) & 1.0119 & 0.9972 & 0.9956 & 0.9972 &  & 1.0118 & \textbf{0.8865} & 1.041 & \textbf{0.9073} \\
        (1, -2) & 1.0091 & 0.9973 & 1.0043 & 0.9976 &  & 1.009 & \textbf{0.8866} & 1.072 & \textbf{0.903} \\
        (2, -1) & 1.01 & 1.0002 & 1.0014 & 0.998 &  & 1.0091 & \textbf{0.8894} & \textbf{1.0871} & \textbf{0.8974} \\
        \hline
        (-2, -2) & 1.0083 & 1.0001 & 1.0038 & 0.9975 &  & 1.0082 & \textbf{0.8895} & 1.0163 & \textbf{0.9018} \\
        (-1, -1) & 1.0107 & 0.999 & 1.0039 & 0.9971 &  & 1.0106 & \textbf{0.8882} & 1.0305 & \textbf{0.9064} \\
        (0, 0) & 1.0116 & 0.9968 & 1.0005 & 0.9982 &  & 1.0118 & \textbf{0.8866} & 1.0453 & \textbf{0.9097} \\
        (1, 1) & 1.0096 & 0.9955 & 0.9981 & 0.9995 &  & 1.0095 & \textbf{0.8862} & \textbf{1.0619} & \textbf{0.909} \\
        (2, 2) & 1.0101 & 0.9994 & 0.9968 & 1.0011 &  & 1.0098 & \textbf{0.8889} & 1.0802 & \textbf{0.9066} \\
        (-2, 0) & 1.0107 & 0.9989 & 1.0008 & 0.9991 &  & 1.0113 & \textbf{0.8886} & 1.0156 & \textbf{0.8998} \\
        (-1, 1) & 1.0113 & 0.9983 & 0.9989 & 0.9995 &  & 1.011 & \textbf{0.8874} & 1.0262 & \textbf{0.9055} \\
        (0, 2) & 1.0116 & 0.9969 & 0.9967 & 0.9976 &  & 1.0115 & \textbf{0.8869} & 1.042 & \textbf{0.9081} \\
        (1, -2) & 1.0093 & 0.9975 & 1.0043 & 0.999 &  & 1.0093 & \textbf{0.8869} & 1.0721 & \textbf{0.9041} \\
        (2, -1) & 1.0098 & 0.9996 & 1.0024 & 0.9975 &  & 1.0097 & \textbf{0.8891} & \textbf{1.0887} & \textbf{0.8974} \\
        \hline
        (-2, -2) & 1.0096 & 0.9977 & 1.0034 & 0.9979 &  & 1.0095 & \textbf{0.8875} & 1.0162 & \textbf{0.9027} \\
        (-1, -1) & 1.0094 & 0.9978 & 1.0038 & 0.9976 &  & 1.0097 & \textbf{0.8872} & 1.0304 & \textbf{0.9067} \\
        (0, 0) & 1.0106 & 0.9963 & 1.0013 & 0.9974 &  & 1.0104 & \textbf{0.8859} & 1.0459 & \textbf{0.9085} \\
        (1, 1) & 1.0085 & 0.9965 & 0.9982 & 1.0001 &  & 1.0084 & \textbf{0.8857} & \textbf{1.0621} & \textbf{0.9093} \\
        (2, 2) & 1.0079 & 0.9983 & 0.9961 & 0.9983 &  & 1.0078 & \textbf{0.8879} & 1.0786 & \textbf{0.9043} \\
        (-2, 0) & 1.0099 & 0.9983 & 1.0007 & 0.9977 &  & 1.01 & \textbf{0.8872} & 1.0155 & \textbf{0.8988} \\
        (-1, 1) & 1.0112 & 0.9976 & 0.9983 & 0.9981 &  & 1.011 & \textbf{0.8866} & 1.0262 & \textbf{0.9043} \\
        (0, 2) & 1.0108 & 0.9962 & 0.9959 & 0.9957 &  & 1.0108 & \textbf{0.8863} & 1.0411 & \textbf{0.906} \\
        (1, -2) & 1.0078 & 0.9973 & 1.0034 & 0.998 &  & 1.0074 & \textbf{0.8864} & 1.071 & \textbf{0.9031} \\
        (2, -1) & 1.0082 & 0.9996 & 1.0036 & 0.9969 &  & 1.0079 & \textbf{0.8887} & \textbf{1.0881} & \textbf{0.8972} \\
        \hline
    \end{tabular}
    \label{tb:multistage_one}
\end{table}

\begin{table}[ht]
    \caption{Median parameter estimates investigating the impact of treatment thresholds in a multistage DTR, by varying ($\psi_{11}, \psi_{21}$) as indicated. Blip parameter estimates are compared for $n=10000$ individuals, using the corrected method compared to a naive analysis. The top set of rows of the table use the first error-prone proxy at both stages, the second set of rows use the mean of proxies at both stages, the third set of rows use the mean at the first stage and the first error-prone proxy at the second, and the final set of rows use the first error-prone proxy at the first stage and the mean at the second. Bold values indicate parameters for which the $95\%$ percentile-based interval across the $1000$ simulation replicates did \textbf{not} cover the true parameter value.}
    \centering
    \begin{tabular}{lccccccccc}
        \toprule & \multicolumn{4}{c}{Regression Calibration} & & \multicolumn{4}{c}{Naive}\\
        $(\psi_{10}, \psi_{11}, \psi_{20}, \psi_{21})$ & $A_{1}$ & $A_1X_{1}$ & $A_{2}$ & $A_2X_{2}$ & & $A_{1}$ & $A_1X_{1}$ & $A_{2}$ & $A_2X_{2}$ \\
        \toprule 
        (1, -1, 1, -1) & 0.9902 & -1.001 & 1.0011 & -1.002 &  & 0.9907 & \textbf{-0.8904} & 0.956 & \textbf{-0.9114} \\
        (1, -0.1, 1, -0.1) & 0.9999 & -0.1017 & 1.0017 & -0.1006 &  & 0.9998 & -0.0904 & 0.9971 & -0.0914 \\
        (1, 0, 1, 0) & 1 & -0.0019 & 1.0014 & -8e-04 &  & 1 & -0.0017 & 1.0016 & -9e-04 \\
        (1, 0.1, 1, 0.1) & 0.9998 & 0.0979 & 1.0016 & 0.099 &  & 0.9998 & 0.0872 & 1.006 & 0.0904 \\
        (1, 1, 1, 1) & 1.0106 & 0.9961 & 1.0003 & 0.9995 &  & 1.0108 & \textbf{0.8859} & 1.045 & \textbf{0.9101} \\
        (1, -1, 1, 0) & 1.0006 & -1.0014 & 1.0017 & -1e-04 &  & 1.0005 & \textbf{-0.89} & 1.0017 & 2e-04 \\
        (1, -0.1, 1, 0.1) & 0.9998 & -0.1018 & 1.0017 & 0.0994 &  & 0.9998 & -0.0905 & 1.0059 & 0.0906 \\
        (1, 0, 1, 1) & 1.0106 & -0.0029 & 1.0017 & 0.9986 &  & 1.0105 & -0.0026 & 1.0461 & \textbf{0.9094} \\
        (1, 0.1, 1, -1) & 0.9918 & 0.0988 & 1.0019 & -1.0028 &  & 0.9918 & 0.0877 & 0.9563 & \textbf{-0.9121} \\
        (1, 1, 1, -0.1) & 0.9998 & 0.9967 & 1.0002 & -0.1017 &  & 1.0004 & \textbf{0.8867} & 0.9958 & -0.0922 \\
        \hline
        (1, -1, 1, -1) & 0.993 & -1.0006 & 0.9997 & -1.0023 &  & 0.9934 & \textbf{-0.8895} & 0.9548 & \textbf{-0.911} \\
        (1, -0.1, 1, -0.1) & 1.0001 & -0.1017 & 0.9993 & -0.1018 &  & 1.0001 & -0.0903 & 0.9951 & -0.0924 \\
        (1, 0, 1, 0) & 1.0006 & -0.0017 & 0.9995 & -0.0017 &  & 1.0006 & -0.0015 & 0.9996 & -0.0018 \\
        (1, 0.1, 1, 0.1) & 1.0007 & 0.0985 & 0.9999 & 0.0976 &  & 1.0007 & 0.0875 & 1.0046 & 0.0889 \\
        (1, 1, 1, 1) & 1.0111 & 0.997 & 1.0013 & 0.9967 &  & 1.0111 & \textbf{0.8872} & 1.0454 & \textbf{0.9074} \\
        (1, -1, 1, 0) & 1.0025 & -1.0015 & 0.9996 & -0.0015 &  & 1.0025 & \textbf{-0.8903} & 0.9994 & -0.0012 \\
        (1, -0.1, 1, 0.1) & 1.0006 & -0.1016 & 0.9993 & 0.0981 &  & 1.0006 & -0.0903 & 1.004 & 0.0895 \\
        (1, 0, 1, 1) & 1.0117 & -0.0015 & 1.0016 & 0.9968 &  & 1.0117 & -0.0013 & 1.0459 & \textbf{0.9076} \\
        (1, 0.1, 1, -1) & 0.9924 & 0.0992 & 1.0005 & -1.0034 &  & 0.9925 & 0.0881 & 0.9549 & \textbf{-0.9123} \\
        (1, 1, 1, -0.1) & 0.9997 & 0.9969 & 1.0006 & -0.1023 &  & 1.0001 & \textbf{0.8871} & 0.9957 & -0.0933 \\
        \hline
        (1, -1, 1, -1) & 0.9926 & -1.0004 & 1.0011 & -1.0041 &  & 0.9927 & \textbf{-0.8893} & 0.9556 & \textbf{-0.913} \\
        (1, -0.1, 1, -0.1) & 1.0015 & -0.1013 & 1 & -0.1008 &  & 1.0015 & -0.09 & 0.9959 & -0.092 \\
        (1, 0, 1, 0) & 1.0015 & -0.0017 & 1.0001 & -7e-04 &  & 1.0015 & -0.0015 & 1 & -8e-04 \\
        (1, 0.1, 1, 0.1) & 1.0018 & 0.0984 & 1.0002 & 0.0991 &  & 1.0017 & 0.0875 & 1.0046 & 0.0902 \\
        (1, 1, 1, 1) & 1.0116 & 0.9968 & 1.0005 & 0.9982 &  & 1.0118 & \textbf{0.8866} & 1.0453 & \textbf{0.9097} \\
        (1, -1, 1, 0) & 1.0023 & -1.0016 & 1.0008 & -0.0011 &  & 1.0022 & \textbf{-0.8901} & 1.0008 & -0.0015 \\
        (1, -0.1, 1, 0.1) & 1.0016 & -0.1016 & 1.0006 & 0.0994 &  & 1.0016 & -0.0903 & 1.0048 & 0.09 \\
        (1, 0, 1, 1) & 1.0121 & -0.0014 & 1.0004 & 0.9982 &  & 1.0121 & -0.0012 & 1.045 & \textbf{0.9097} \\
        (1, 0.1, 1, -1) & 0.9926 & 0.0996 & 0.9999 & -1.0028 &  & 0.9927 & 0.0885 & 0.9549 & \textbf{-0.9121} \\
        (1, 1, 1, -0.1) & 1.001 & 0.9971 & 0.9998 & -0.1007 &  & 1.0009 & \textbf{0.8872} & 0.9952 & -0.0919 \\
        \hline
        (1, -1, 1, -1) & 0.9893 & -1.0008 & 1.0004 & -1.0018 &  & 0.9893 & \textbf{-0.8897} & 0.9554 & \textbf{-0.9112} \\
        (1, -0.1, 1, -0.1) & 0.9992 & -0.1017 & 1.0013 & -0.1014 &  & 0.9992 & -0.0906 & 0.9964 & -0.0928 \\
        (1, 0, 1, 0) & 0.9994 & -0.0019 & 1.0013 & -0.0016 &  & 0.9993 & -0.0017 & 1.001 & -0.0018 \\
        (1, 0.1, 1, 0.1) & 0.9997 & 0.0979 & 1.0012 & 0.0982 &  & 0.9996 & 0.0873 & 1.0058 & 0.0893 \\
        (1, 1, 1, 1) & 1.0106 & 0.9963 & 1.0013 & 0.9974 &  & 1.0104 & \textbf{0.8859} & 1.0459 & \textbf{0.9085} \\
        (1, -1, 1, 0) & 1.0008 & -1.0013 & 0.9997 & -6e-04 &  & 1.0009 & \textbf{-0.8901} & 0.9996 & -8e-04 \\
        (1, -0.1, 1, 0.1) & 0.9997 & -0.1014 & 1.001 & 0.0987 &  & 0.9996 & -0.0903 & 1.0056 & 0.0894 \\
        (1, 0, 1, 1) & 1.0102 & -0.0023 & 1.0011 & 0.9986 &  & 1.0102 & -0.002 & 1.0456 & \textbf{0.9091} \\
        (1, 0.1, 1, -1) & 0.9914 & 0.0986 & 1.0013 & -1.0021 &  & 0.9914 & 0.0876 & 0.9563 & \textbf{-0.9119} \\
        (1, 1, 1, -0.1) & 0.9994 & 0.9966 & 1.0005 & -0.1016 &  & 0.9995 & \textbf{0.8867} & 0.9958 & -0.0925 \\
        \hline
    \end{tabular}
    \label{tb:multistage_two}
\end{table}

\begin{table}[ht]
    \caption{Median parameter estimates investigating the impact of treatment probabilities in a multistage DTR, by varying the true treatment-free model as indicated. Linear treatment-free models are used in all settings. Blip parameter estimates are compared for $n=10000$ individuals, using the corrected method compared to a naive analysis. The top set of rows of the table use the first error-prone proxy at both stages, the second set of rows use the mean of proxies at both stages, the third set of rows use the mean at the first stage and the first error-prone proxy at the second, and the final set of rows use the first error-prone proxy at the first stage and the mean at the second. Bold values indicate parameters for which the $95\%$ percentile-based interval across the $1000$ simulation replicates did \textbf{not} cover the true parameter value.}
    \centering
    \begin{tabular}{lccccccccc}
        \toprule & \multicolumn{4}{c}{Regression Calibration} & & \multicolumn{4}{c}{Naive}\\
        Treatment-Free Model & $A_1$ & $A_1X_{1}$ & $A_{2}$ & $A_2X_{2}$ & & $A_{1}$ & $A_1X_{1}$ & $A_{2}$ & $A_2X_{2}$ \\
        \toprule 
        Linear & 1.0106 & 0.9961 & 1.0003 & 0.9995 &  & 1.0108 & \textbf{0.8859} & 1.045 & \textbf{0.9101} \\
        Quadratic & 1.0105 & 1.0043 & 1.0012 & 0.9982 &  & 1.0104 & \textbf{0.893} & 1.0457 & \textbf{0.9098} \\
        Cubic & 1.0019 & 1.0108 & 1.0041 & 0.9993 &  & 1.0021 & 0.8986 & 1.0484 & 0.9099 \\
        Exponential & 1.0042 & 1.0105 & 1.0035 & 0.9986 &  & 1.0041 & 0.8976 & 1.0476 & 0.9104 \\
        Complex & 1.0125 & 1.0038 & 0.9994 & 1.0007 &  & 1.0127 & 0.8921 & 1.044 & \textbf{0.9116} \\
        \hline
        Linear & 1.0111 & 0.997 & 1.0013 & 0.9967 &  & 1.0111 & \textbf{0.8872} & 1.0454 & \textbf{0.9074} \\
        Quadratic & 1.0106 & 1.0017 & 1.0033 & 0.9959 &  & 1.0107 & \textbf{0.8901} & 1.0483 & \textbf{0.9078} \\
        Cubic & 1.0069 & 1.0043 & 1.0072 & 1.0012 &  & 1.0068 & 0.8938 & 1.0525 & 0.9129 \\
        Exponential & 1.0089 & 1.0066 & 1.0054 & 1.0015 &  & 1.0091 & 0.8939 & 1.0503 & 0.9122 \\
        Complex & 1.0113 & 1.0014 & 1.0035 & 0.9969 &  & 1.0116 & \textbf{0.8903} & 1.0467 & \textbf{0.9082} \\
        \hline
        Linear & 1.0116 & 0.9968 & 1.0005 & 0.9982 &  & 1.0118 & \textbf{0.8866} & 1.0453 & \textbf{0.9097} \\
        Quadratic & 1.0116 & 1.001 & 1.0005 & 0.9985 &  & 1.0115 & \textbf{0.8897} & 1.044 & \textbf{0.9093} \\
        Cubic & 1.0075 & 1.0034 & 1.0038 & 0.9967 &  & 1.0075 & 0.8937 & 1.0484 & 0.9095 \\
        Exponential & 1.0091 & 1.0069 & 1.0033 & 0.9996 &  & 1.0087 & 0.8949 & 1.0477 & 0.9107 \\
        Complex & 1.0119 & 1.0026 & 1 & 1.0012 &  & 1.0116 & \textbf{0.8904} & 1.0445 & \textbf{0.9108} \\
        \hline
        Linear & 1.0106 & 0.9963 & 1.0013 & 0.9974 &  & 1.0104 & \textbf{0.8859} & 1.0459 & \textbf{0.9085} \\
        Quadratic & 1.01 & 1.004 & 1.003 & 0.997 &  & 1.0104 & \textbf{0.8927} & 1.0472 & \textbf{0.9078} \\
        Cubic & 1.0017 & 1.0101 & 1.0044 & 0.9992 &  & 1.0021 & 0.8974 & 1.0504 & 0.9104 \\
        Exponential & 1.0046 & 1.0097 & 1.0061 & 0.9996 &  & 1.0045 & 0.8975 & 1.0503 & 0.9104 \\
        Complex & 1.0115 & 1.0039 & 1.0015 & 0.9985 &  & 1.0115 & 0.8926 & 1.0457 & \textbf{0.9091} \\
        \hline
    \end{tabular}
    \label{tb:multistage_three}
 \end{table}

\begin{table}[ht]
    \caption{Median parameter estimates investigating the impact of treatment probabilities in a multistage DTR, by varying the treatment models as indicated. Linear treatment models are used in all situations. Blip parameter estimates are compared for $n=10000$ individuals, using the corrected method compared to a naive analysis. The top set of rows of the table use the first error-prone proxy at both stages, the second set of rows use the mean of proxies at both stages, the third set of rows use the mean at the first stage and the first error-prone proxy at the second, and the final set of rows use the first error-prone proxy at the first stage and the mean at the second. Bold values indicate parameters for which the $95\%$ percentile-based interval across the $1000$ simulation replicates did \textbf{not} cover the true parameter value.}
    \centering
    \begin{tabular}{lccccccccc}
        \toprule & \multicolumn{4}{c}{Regression Calibration} & & \multicolumn{4}{c}{Naive}\\
        Treatment Models & $A_1$ & $A_1X_{1}$ & $A_{2}$ & $A_2X_{2}$ & & $A_{1}$ & $A_1X_{1}$ & $A_{2}$ & $A_2X_{2}$ \\
        \toprule 
        Linear/Linear & 1.0106 & 0.9961 & 1.0003 & 0.9995 &  & 1.0108 & \textbf{0.8859} & 1.045 & \textbf{0.9101} \\
        Linear/Quadratic & 0.9992 & 0.9992 & \textbf{0.8902} & \textbf{1.1138} &  & 0.9991 & \textbf{0.8881} & \textbf{0.9407} & 1.014 \\
        Linear/Mixed & 0.9925 & 1 & \textbf{0.9022} & \textbf{1.1187} &  & 0.9923 & \textbf{0.8893} & 0.9529 & 1.0174 \\
        Linear/Exponential & 1.0126 & 0.9989 & 0.9836 & 1.016 &  & 1.0124 & \textbf{0.8876} & 1.0291 & 0.9257 \\
        Quadratic/Quadratic & \textbf{0.9047} & \textbf{1.1117} & \textbf{0.8929} & \textbf{1.1153} &  & \textbf{0.9044} & 0.9885 & 0.9611 & 1.0159 \\
        Quadratic/Mixed & \textbf{0.8969} & \textbf{1.1119} & \textbf{0.9011} & \textbf{1.1188} &  & \textbf{0.897} & 0.9886 & 0.9707 & 1.0199 \\
        Quadratic/Exponential & \textbf{0.918} & \textbf{1.1125} & 0.9846 & 1.0189 &  & \textbf{0.918} & 0.9899 & 1.0478 & 0.9289 \\
        Mixed/Mixed & \textbf{0.9033} & \textbf{1.1052} & \textbf{0.9078} & 1.1216 &  & \textbf{0.9033} & 0.9824 & 0.9946 & 1.0192 \\
        Mixed/Exponential & \textbf{0.9246} & \textbf{1.1096} & 0.9853 & 1.0153 &  & \textbf{0.9244} & 0.9872 & 1.0652 & 0.9223 \\
        Exponential/Exponential & 0.9979 & 1.013 & 0.9853 & 1.0131 &  & 0.9983 & \textbf{0.9007} & 1.0479 & 0.9244 \\
        \hline
        Linear/Linear & 1.0111 & 0.997 & 1.0013 & 0.9967 &  & 1.0111 & \textbf{0.8872} & 1.0454 & \textbf{0.9074} \\
        Linear/Quadratic & 1.0013 & 0.9987 & \textbf{0.8847} & \textbf{1.1219} &  & 1.0012 & \textbf{0.8876} & \textbf{0.9353} & 1.0215 \\
        Linear/Mixed & 0.9915 & 0.9979 & \textbf{0.8915} & \textbf{1.1179} &  & 0.9916 & \textbf{0.8868} & 0.9413 & 1.0179 \\
        Linear/Exponential & 1.0128 & 0.999 & 0.9818 & 1.0194 &  & 1.013 & \textbf{0.8882} & 1.0271 & 0.9275 \\
        Quadratic/Quadratic & \textbf{0.898} & \textbf{1.1245} & \textbf{0.8863} & \textbf{1.1239} &  & \textbf{0.8982} & 0.9994 & 0.9547 & 1.0234 \\
        Quadratic/Mixed & \textbf{0.8889} & \textbf{1.1226} & \textbf{0.8943} & \textbf{1.1241} &  & \textbf{0.8887} & 0.9981 & 0.9629 & 1.0257 \\
        Quadratic/Exponential & \textbf{0.9089} & \textbf{1.121} & 0.9809 & 1.021 &  & \textbf{0.9087} & 0.9965 & 1.0429 & 0.932 \\
        Mixed/Mixed & \textbf{0.8941} & \textbf{1.1119} & \textbf{0.9053} & 1.1241 &  & \textbf{0.8938} & 0.9878 & 0.9923 & 1.0228 \\
        Mixed/Exponential & \textbf{0.9142} & \textbf{1.1125} & 0.9834 & 1.0242 &  & \textbf{0.9141} & 0.9884 & 1.0622 & 0.9313 \\
        Exponential/Exponential & 0.9945 & 1.0162 & 0.9831 & 1.0186 &  & 0.9948 & \textbf{0.9021} & 1.0474 & 0.9294 \\
        \hline
        Linear/Linear & 1.0116 & 0.9968 & 1.0005 & 0.9982 &  & 1.0118 & \textbf{0.8866} & 1.0453 & \textbf{0.9097} \\
        Linear/Quadratic & 1.0012 & 0.9979 & \textbf{0.8904} & \textbf{1.1137} &  & 1.0008 & \textbf{0.8876} & \textbf{0.9412} & 1.0139 \\
        Linear/Mixed & 0.9922 & 0.9985 & \textbf{0.8999} & \textbf{1.1146} &  & 0.9923 & \textbf{0.8881} & 0.9495 & 1.0143 \\
        Linear/Exponential & 1.011 & 0.9982 & 0.9823 & 1.0161 &  & 1.0112 & \textbf{0.8869} & 1.0281 & 0.9263 \\
        Quadratic/Quadratic & \textbf{0.8978} & \textbf{1.1239} & \textbf{0.8921} & \textbf{1.1151} &  & \textbf{0.8982} & 0.9981 & 0.9598 & 1.0149 \\
        Quadratic/Mixed & \textbf{0.8904} & \textbf{1.1216} & \textbf{0.9021} & \textbf{1.1175} &  & \textbf{0.8903} & 0.9979 & 0.9703 & 1.0182 \\
        Quadratic/Exponential & \textbf{0.9086} & \textbf{1.1215} & 0.9848 & 1.0112 &  & \textbf{0.9083} & 0.9972 & 1.0464 & 0.923 \\
        Mixed/Mixed & \textbf{0.8932} & \textbf{1.1145} & \textbf{0.9077} & 1.1196 &  & \textbf{0.8937} & 0.9898 & 0.9943 & 1.0184 \\
        Mixed/Exponential & \textbf{0.9133} & \textbf{1.1123} & 0.9854 & 1.0169 &  & \textbf{0.9126} & 0.9888 & 1.0637 & 0.9247 \\
        Exponential/Exponential & 0.996 & 1.0178 & 0.9844 & 1.0135 &  & 0.9955 & \textbf{0.905} & 1.0486 & 0.9253 \\
        \hline
        Linear/Linear & 1.0106 & 0.9963 & 1.0013 & 0.9974 &  & 1.0104 & \textbf{0.8859} & 1.0459 & \textbf{0.9085} \\
        Linear/Quadratic & 1.0003 & 0.9973 & \textbf{0.884} & \textbf{1.1229} &  & 0.9997 & \textbf{0.8868} & \textbf{0.9348} & 1.0231 \\
        Linear/Mixed & 0.9908 & 0.998 & \textbf{0.8919} & \textbf{1.1214} &  & 0.9904 & \textbf{0.8869} & 0.9419 & 1.0213 \\
        Linear/Exponential & 1.0111 & 0.9987 & 0.9803 & 1.0186 &  & 1.0113 & \textbf{0.888} & 1.0264 & 0.9292 \\
        Quadratic/Quadratic & \textbf{0.9049} & \textbf{1.1118} & \textbf{0.8872} & \textbf{1.1245} &  & \textbf{0.9041} & 0.9884 & 0.9563 & 1.0241 \\
        Quadratic/Mixed & \textbf{0.8951} & \textbf{1.1141} & \textbf{0.8978} & 1.1272 &  & \textbf{0.8956} & 0.9901 & 0.9677 & 1.0279 \\
        Quadratic/Exponential & \textbf{0.9178} & \textbf{1.1119} & 0.9811 & 1.0225 &  & \textbf{0.9179} & 0.9884 & 1.0437 & 0.9333 \\
        Mixed/Mixed & \textbf{0.9008} & \textbf{1.1039} & \textbf{0.9066} & 1.1307 &  & \textbf{0.901} & 0.9811 & 0.9946 & 1.0289 \\
        Mixed/Exponential & \textbf{0.9248} & \textbf{1.1063} & 0.9843 & 1.023 &  & \textbf{0.9249} & 0.9835 & 1.0644 & 0.9311 \\
        Exponential/Exponential & 0.9971 & 1.0136 & 0.984 & 1.0232 &  & 0.997 & \textbf{0.9008} & 1.0482 & 0.9345 \\
        \hline
    \end{tabular}
    \label{tb:multistage_four}
\end{table}

\begin{table}[ht]
    \caption{Median parameter estimates investigating the impact of treatment probabilities in a multistage DTR, by varying the error-models as indicated. Blip parameter estimates are compared for $n=10000$ individuals, using the corrected method compared to a naive analysis. The top set of rows of the table use the first error-prone proxy at both stages, the second set of rows use the mean of proxies at both stages, the third set of rows use the mean at the first stage and the first error-prone proxy at the second, and the final set of rows use the first error-prone proxy at the first stage and the mean at the second. Bold values indicate parameters for which the $95\%$ percentile-based interval across the $1000$ simulation replicates did \textbf{not} cover the true parameter value.}
    \centering
    \begin{tabular}{lccccccccc}
        \toprule & \multicolumn{4}{c}{Regression Calibration} & & \multicolumn{4}{c}{Naive}\\
        Error Models & $A_1$ & $A_1X_{1}$ & $A_{2}$ & $A_2X_{2}$ & & $A_{1}$ & $A_1X_{1}$ & $A_{2}$ & $A_2X_{2}$ \\
        \toprule 
        Normal/Normal & 1.0106 & 0.9961 & 1.0003 & 0.9995 &  & 1.0108 & \textbf{0.8859} & 1.045 & \textbf{0.9101} \\
        Normal/Approx. Normal & 1.0175 & 0.9954 & 1.0007 & 0.9975 &  & 1.0174 & \textbf{0.8248} & \textbf{1.0705} & \textbf{0.8586} \\
        Normal/Gamma & 1.0109 & 1.0292 & 1.0036 & 1.0078 &  & 1.0103 & \textbf{0.8577} & \textbf{1.0749} & \textbf{0.8647} \\
        Normal/Uniform & 1.0062 & 1.034 & 1.0028 & 1.0226 &  & 1.0062 & 0.9729 & 1.0339 & 0.9616 \\
        Approx. Normal/Approx. Normal & 1.0428 & 0.996 & 0.9989 & 0.9916 &  & 1.0414 & \textbf{0.6132} & \textbf{1.1612} & \textbf{0.6666} \\
        Approx. Normal/Gamma & 1.0187 & \textbf{1.1453} & 1.0069 & 1.0592 &  & 1.019 & \textbf{0.7388} & \textbf{1.1893} & \textbf{0.7095} \\
        Approx. Normal/Uniform & 1.0109 & 1.0405 & 1.0028 & 1.0318 &  & 1.0107 & 0.9659 & 1.0443 & 0.9495 \\
        Gamma/Gamma & \textbf{1.0866} & \textbf{1.2838} & \textbf{0.919} & \textbf{1.202} &  & \textbf{1.086} & \textbf{0.857} & \textbf{1.1308} & \textbf{0.8097} \\
        Gamma/Uniform & 1.0157 & 1.0608 & 0.9766 & 1.0494 &  & 1.0156 & 0.9842 & 1.0198 & 0.9638 \\
        Uniform/Uniform & 1.0075 & 1.0367 & 0.991 & 1.043 &  & 1.0076 & 0.9957 & 1.0152 & 0.9957 \\
        \hline
        Normal/Normal & 1.0111 & 0.997 & 1.0013 & 0.9967 &  & 1.0111 & \textbf{0.8872} & 1.0454 & \textbf{0.9074} \\
        Normal/Approx. Normal & 1.0173 & 0.9968 & 1.0002 & 0.9994 &  & 1.0165 & \textbf{0.8249} & \textbf{1.0706} & \textbf{0.8586} \\
        Normal/Gamma & 1.0247 & 1.0323 & 0.9832 & 1.0201 &  & 1.0248 & \textbf{0.8602} & \textbf{1.0552} & \textbf{0.877} \\
        Normal/Uniform & 1.0103 & 1.0358 & 0.995 & 1.0339 &  & 1.0098 & 0.9747 & 1.0269 & 0.9718 \\
        Approx. Normal/Approx. Normal & 1.0393 & 1.008 & 0.9994 & 0.9984 &  & 1.0384 & \textbf{0.6209} & \textbf{1.163} & \textbf{0.6743} \\
        Approx. Normal/Gamma & 1.0537 & \textbf{1.1959} & 0.9615 & \textbf{1.1174} &  & 1.0534 & \textbf{0.7688} & \textbf{1.1499} & \textbf{0.7544} \\
        Approx. Normal/Uniform & 1.0104 & 1.0498 & 0.9957 & 1.0521 &  & 1.0101 & 0.9743 & 1.0383 & 0.9692 \\
        Gamma/Gamma & \textbf{1.0951} & \textbf{1.4009} & \textbf{0.9042} & \textbf{1.3366} &  & \textbf{1.0947} & 0.9351 & \textbf{1.1268} & 0.952 \\
        Gamma/Uniform & 1.0161 & 1.0605 & 0.9816 & 1.0662 &  & 1.0159 & 0.9859 & 1.025 & 0.9803 \\
        Uniform/Uniform & 1.0079 & 1.033 & 0.9914 & 1.0418 &  & 1.0077 & 0.9923 & 1.0154 & 0.9942 \\
        \hline
        Normal/Normal & 1.0116 & 0.9968 & 1.0005 & 0.9982 &  & 1.0118 & \textbf{0.8866} & 1.0453 & \textbf{0.9097} \\
        Normal/Approx. Normal & 1.0177 & 0.9966 & 1.0008 & 0.9977 &  & 1.0171 & \textbf{0.8257} & \textbf{1.0708} & \textbf{0.8578} \\
        Normal/Gamma & 1.0263 & 1.0326 & 1.0028 & 1.0074 &  & 1.0262 & \textbf{0.86} & \textbf{1.0744} & \textbf{0.8649} \\
        Normal/Uniform & 1.0112 & 1.036 & 1.0029 & 1.0235 &  & 1.0112 & 0.9746 & 1.0344 & 0.9619 \\
        Approx. Normal/Approx. Normal & 1.0372 & 1.0098 & 0.9988 & 0.9909 &  & 1.037 & \textbf{0.6207} & \textbf{1.162} & \textbf{0.6673} \\
        Approx. Normal/Gamma & 1.0495 & \textbf{1.1943} & 1.007 & 1.058 &  & 1.0493 & \textbf{0.7689} & \textbf{1.1893} & \textbf{0.7076} \\
        Approx. Normal/Uniform & 1.0116 & 1.0498 & 1.0027 & 1.0313 &  & 1.0117 & 0.9742 & 1.0453 & 0.9487 \\
        Gamma/Gamma & \textbf{1.072} & \textbf{1.4024} & \textbf{0.9193} & \textbf{1.2012} &  & \textbf{1.0719} & 0.9346 & \textbf{1.1294} & \textbf{0.809} \\
        Gamma/Uniform & 1.0161 & 1.0601 & 0.9755 & 1.0494 &  & 1.016 & 0.9857 & 1.0197 & 0.9645 \\
        Uniform/Uniform & 1.0078 & 1.0329 & 0.9916 & 1.0433 &  & 1.0078 & 0.9924 & 1.0156 & 0.9957 \\
        \hline
        Normal/Normal & 1.0106 & 0.9963 & 1.0013 & 0.9974 &  & 1.0104 & \textbf{0.8859} & 1.0459 & \textbf{0.9085} \\
        Normal/Approx. Normal & 1.017 & 0.9955 & 1.0009 & 0.9985 &  & 1.0166 & \textbf{0.8249} & \textbf{1.0704} & \textbf{0.8581} \\
        Normal/Gamma & 1.0095 & 1.0288 & 0.9834 & 1.0218 &  & 1.0089 & \textbf{0.8568} & \textbf{1.0569} & \textbf{0.8763} \\
        Normal/Uniform & 1.0047 & 1.034 & 0.9952 & 1.0332 &  & 1.0048 & 0.9736 & 1.0266 & 0.9711 \\
        Approx. Normal/Approx. Normal & 1.0441 & 0.997 & 0.9987 & 0.9996 &  & 1.0439 & \textbf{0.6129} & \textbf{1.1602} & \textbf{0.6746} \\
        Approx. Normal/Gamma & 1.0232 & \textbf{1.1462} & 0.9612 & \textbf{1.1165} &  & 1.023 & \textbf{0.7378} & \textbf{1.1498} & \textbf{0.7538} \\
        Approx. Normal/Uniform & 1.0101 & 1.0408 & 0.9958 & 1.0529 &  & 1.0102 & 0.9655 & 1.0377 & 0.969 \\
        Gamma/Gamma & \textbf{1.1095} & \textbf{1.2825} & \textbf{0.904} & \textbf{1.337} &  & \textbf{1.1087} & \textbf{0.8556} & \textbf{1.1267} & 0.9509 \\
        Gamma/Uniform & 1.016 & 1.0603 & 0.9816 & 1.0659 &  & 1.0159 & 0.9845 & 1.0256 & 0.9818 \\
        Uniform/Uniform & 1.0074 & 1.0369 & 0.9916 & 1.0424 &  & 1.0074 & 0.9957 & 1.0157 & 0.9946 \\
        \hline
    \end{tabular}
    \label{tb:multistage_five}
\end{table}

\end{document}